\documentclass[a4paper]{article}
\usepackage[T1]{fontenc}
\usepackage{lmodern}

\usepackage{a4wide}
\usepackage{graphicx}
\usepackage{subfigure}
\usepackage{booktabs,array}
\usepackage[round]{natbib}

\usepackage{amsthm,amsmath,amsfonts}
\usepackage{xspace}
\usepackage{color}
\usepackage[textsize=small,textwidth=4cm,color=green]{todonotes}
\usepackage{tikz, url}
\usepackage[breaklinks]{hyperref}
\hypersetup{colorlinks=true}
\hypersetup{linkcolor=black,anchorcolor=black,citecolor=black,urlcolor=black}
\hypersetup{pdftitle=Online metric algorithms with untrusted predictions}
\hypersetup{pdfauthor={Antonios Antoniadis} {Christian Coester} {Marek Elias} {Adam Polak} {Bertrand Simon}}
\hypersetup{pdfkeywords={metrical task systems} {online algorithms} {competitive analysis}}
\usepackage{cleveref}
\usepackage{verbatim}
\usepackage{booktabs}
\usetikzlibrary{shapes.misc}
\usepackage[algo2e,ruled,vlined]{algorithm2e}

\newtheorem{definition}{Definition}
\newtheorem{theorem}{Theorem}
\newtheorem*{theorem*}{Theorem}
\newtheorem{lemma}[theorem]{Lemma}

\newtheorem{remark}[theorem]{Remark}
\newtheorem{claim}[theorem]{Claim}
\newtheorem{restate}{Theorem}

\DeclareMathOperator{\opt}{\textsc{Opt}}
\newcommand{\dist}{\ensuremath{\mathit{dist}}\xspace}
\newcommand{\cost}{\ensuremath{\mathit{cost}}\xspace}
\newcommand{\alg}{\ensuremath{\mathcal{A}}\xspace}
\newcommand{\ftp}{\textsc{FtP}\xspace}
\newcommand{\off}{\textsc{Off}\xspace}
\newcommand{\R}{\ensuremath{\mathbb{R}}\xspace}
\newcommand{\N}{\ensuremath{\mathbb{N}}}
\newcommand{\E}{\ensuremath{\mathbb{E}}}

\newcommand{\algname}{\textsc{Trust\&Doubt}\xspace}

\AtBeginDocument{%
  \providecommand\BibTeX{{%
    \normalfont B\kern-0.5em{\scshape i\kern-0.25em b}\kern-0.8em\TeX}}}

\begin{document}

\title{Online metric algorithms with untrusted predictions\thanks{The authors would like to thank IGAFIT for the organization of the AlgPiE workshop which made this project possible.}}

\author{
Antonios Antoniadis\thanks{Supported by DFG Grant AN 1262/1-1.}\\
{\small University of Twente}
\and
Christian Coester\thanks{Supported by NWO VICI grant 639.023.812 of Nikhil Bansal.}\\
{\small University of Oxford}
\and
Marek Eli\'a\v{s}\thanks{Supported by ERC Starting Grant 759471 of Michael Kapralov.}\\
{\small Bocconi University}
\and
Adam Polak\thanks{Supported by National Science Center of Poland grant 2017/27/N/ST6/01334 and by Swiss National Science Foundation within the project \emph{Lattice Algorithms and Integer Programming} (185030).}\\
{\small Max Planck Institute for Informatics}\\
{\small Jagiellonian University}
\and
Bertrand Simon\thanks{Supported by the DFG Project Number 146371743 -- TRR 89 Invasive Computing.}\\
{\small IN2P3 Computing Center, CNRS}
}

\date{}

\maketitle

\begin{abstract}
  Machine-learned predictors, 
  although achieving very good results for inputs resembling
  training data, cannot possibly provide perfect
  predictions in all situations.  Still, decision-making systems
  that are based on such predictors need not only benefit from good
  predictions, but should also achieve a decent
  performance when the predictions are inadequate. 

  In this paper, we propose a prediction setup for arbitrary \emph{metrical task systems (MTS)}
  (e.g.,~\emph{caching},
  \emph{$k$-server} and \emph{convex body chasing}) and \emph{online matching on the line}.
  We utilize results from the theory of online algorithms to show how to make
  the setup robust.
  Specifically for caching, we present an algorithm whose performance, as a function of the prediction error, is exponentially better than what is achievable for general MTS.
  Finally, we present an empirical evaluation of our methods on real world datasets,
  which suggests practicality.
\end{abstract}

\section{Introduction}

Metrical task systems (MTS), introduced by \citet{BorodinLS92},
are a rich class containing several fundamental problems in online
optimization as special cases, including {\em caching}, {\em $k$-server},
{\em convex body chasing}, and {\em convex function chasing}.
MTS are
capable of modeling many problems arising in computing and production systems
\citep{SleatorT85,ManasseMS90},
movements of service vehicles \citep{DehghaniEHLS17,CoesterK19},
power management of embedded systems as well as data centers
\citep{IraniSG03,LinWAT13},
and are also related to the {\em experts} problem
in online learning, see \citep{DanielyM19,BlumB00}.

Initially, we are given a metric space $M$ of {\em states},
which can be interpreted for example as actions, investment strategies, or configurations
of some production machine.
We start at a predefined initial state $x_0$.
At each time $t = 1, 2, \dotsc$, we are presented with
a {\em cost function} $\ell_t\colon M \to \R^+ \cup \{0,+\infty\}$ and our task is to decide
either to stay at $x_{t-1}$ and pay the cost $\ell_t(x_{t-1})$, or to move
to some other (possibly cheaper) state $x_t$ and pay
$\dist(x_{t-1}, x_t) + \ell_t(x_t)$, where $\dist(x_{t-1}, x_t)$ is the cost
of the transition between states $x_{t-1}$ and $x_t$. 
The objective is to minimize the overall cost incurred over time.

Given that MTS is an online problem, one needs to make each decision without any information about the future
cost functions. This makes the problem substantially difficult, as
supported by strong lower bounds for general MTS
\citep{BorodinLS92} as well as for many special MTS problems
\citep[see e.g.][]{KarloffRR94, FiatFKRRV98}.
For the recent work on MTS, see \citet{BubeckCLL19,CoesterL19,BubeckR19}.

In this paper, we study how to utilize predictors (possibly based on
machine learning) in order to decrease the uncertainty about the future
and achieve a better performance for MTS.
We propose a natural prediction setup for MTS and show how to develop
algorithms in this setup with the following properties of
\emph{consistency} (i), \emph{smoothness} (ii),
and \emph{robustness} (iii).
\begin{enumerate}
\item[(i)] Their performance with perfect predictions is close to optimal.
\item[(ii)] With decreasing accuracy of the predictions,
their performance deteriorates smoothly as a function of the prediction error.
\item[(iii)] When given poor predictions, their performance is comparable
to that of the best online algorithm which does not use predictions.
\end{enumerate}

Caching and weighted caching problems, which are special cases of MTS,
have already been studied
in this context of utilizing
predictors \cite{LykourisV18,Rohatgi20,AlexWei20,JiangPS20,BansalCKPV20}.
However, the corresponding prediction setups do not seem applicable to
general MTS. For example,
algorithms by \citet{LykourisV18}
and \citet{Rohatgi20} provide similar guarantees by using predictions
of the next reoccurrence time of the current page in the input sequence.
However, as we show in this paper, such predictions are not useful for more
general MTS: even for weighted caching, they do not help to improve upon the
bounds achievable without predictions unless additional assumptions
are made (see \cite{BansalCKPV20} for an example of such an assumption).

We propose a prediction setup based on {\em action predictions}
where, at each time step, the predictor tries to predict the action that an offline algorithm would have taken. We can view these predictions as recommendations of what our algorithm should do.
We show that using this prediction setup, we can achieve
consistency, smoothness, and robustness for any MTS.
For the (unweighted) caching problem, we develop an algorithm that
obtains a better dependency on the prediction error than our general
result, and whose performance in empirical tests 
is either better or comparable to the algorithms by
\citet{LykourisV18} and \citet{Rohatgi20}.
This demonstrates the flexibility of our setup.
We would like to stress
that specifically for the caching problem, the action predictions can be
obtained by simply converting the reoccurrence time predictions used in
\citep{LykourisV18,Rohatgi20,AlexWei20}, a feature that we use in order to
compare our results to those previous algorithms.
Nevertheless our
prediction setup is applicable to the much broader context of MTS. We
demonstrate this and suggest practicability of our algorithms also for MTS other than caching by providing experimental results for the {\em
  ice cream} problem~\citep{chrobak98}, a simple example of an MTS. Finally, we extend our theoretical result beyond MTS to {\em
  online matching on the line}.

\paragraph{Action Predictions for MTS.}
At each time $t$, the predictor proposes an {\em action}, i.e., a state $p_t$ in the metric space $M$.
We define the \emph{prediction error} with respect to some offline algorithm \off as
\begin{equation}\label{eq:eta-def}
\eta = \sum_{t=1}^T \eta_t;\quad \eta_t = \dist(p_t, o_t),
\end{equation}
where $o_t$ denotes the state of \off at time $t$ and $T$ denotes the length
of the input sequence.

The predictions could be, for instance, the output of a machine-learned model
or actions of a heuristic
which tends to produce good solutions in practice, but possibly without a theoretical guarantee. The offline algorithm \off can be an optimal one, but also other options are plausible.
For example, if the typical instances are
composed of subpatterns known from the past and for which
good solutions are known, then \off could also be a near-optimal
algorithm which composes its output from the partial solutions to the
subpatterns.
The task of the predictor in this case is to anticipate which subpattern
is going to follow and provide the precomputed solution to that
subpattern.
In the case of the caching problem, as mentioned above and
explained in
Section~\ref{sec:setup_prop}, we can actually convert the reoccurrence predictions
\citep{LykourisV18,Rohatgi20,AlexWei20} into action predictions.

Note that,
even if the prediction error with respect to \off is low,
the cost of the solution composed from the predictions $p_1, \dotsc, p_T$
can be much higher than the cost incurred by \off, since $\ell_t(p_t)$
can be much larger than $\ell_t(o_t)$ even if $\dist(p_t,o_t)$ is small.
However, we can design algorithms which use such predictions
and achieve a good performance
whenever the predictions have small error with respect to any
low-cost offline algorithm.
We aim at expressing the performance of the prediction-based algorithms
as a function of $\eta/\off$,
where (abusing notation) $\off$ denotes the cost of the offline
algorithm. This is to avoid scaling issues: if the offline algorithm
incurs movement cost 1000, predictions with total error $\eta=1$ give us a
rather precise estimate of its state, unlike when $\off = 0.1$.

\paragraph{Caching Problem.} In the caching problem we have a two-level computer memory, out of which the
fast one (cache) can only store $k$ pages. We need to answer a
sequence of requests to pages. Such a request requires no action
and incurs no cost if the page is already in the cache, but otherwise a \emph{page fault} occurs and
 we have to add the page and evict some other page at a cost of
$1$. Caching can be seen as an MTS whose states are cache
configurations\footnote{A cache configuration is the set of pages in cache.}. Therefore, also the predictions are cache configurations in our setup, but as we discuss in Section~\ref{sec:setup_prop} they can be encoded much more succinctly than specifying the full cache content in each time step. The error $\eta_t$ describes in this case the number of pages on which the predicted cache and the cache of $\off$ differ at time $t$.

\subsection{Our results}
We prove two general theorems providing robustness and consistency guarantees for any MTS.

\begin{theorem}\label{thm:ub-det}
Let $A$ be a deterministic $\alpha$-competitive online algorithm for a problem $P$
belonging to MTS. Given action predictions for $P$, there is a deterministic algorithm achieving competitive ratio
\[ 9 \cdot \min\{\alpha,\; 1+4\eta/\off\} \]
against any offline algorithm $\off$, where $\eta$ is the prediction error with respect to $\off$.
\end{theorem}
Roughly speaking, the competitive ratio (formally defined in Section~\ref{sec:prelim}) is the worst case ratio between the cost of two algorithms. If $\off$ is an optimal algorithm, then the expression in the theorem is the overall competitive ratio of the prediction-based algorithm.

\begin{theorem}\label{thm:ub-rand}
Let $A$ be a randomized $\alpha$-competitive online algorithm for an MTS $P$
with metric space diameter $D$.
For any $\epsilon \leq 1/4$, given action predictions for $P$ there is a randomized algorithm achieving cost
\[ (1+\epsilon) \cdot \min\{\alpha, 1+4\eta/\off\}\cdot\off + O(D/\epsilon), \]
where $\eta$ is the prediction error with respect to an offline algorithm \off.
Thus, if $\off$ is \mbox{(near-)}optimal and
$\eta \ll \off$, the competitive ratio is close to $1+\epsilon$.
\end{theorem}

We note that the proofs of these theorems are based on the powerful results by
\citet{FiatRR94} and \citet{BlumB00}.
In Theorem~\ref{thm:generalLB},
we show that
the dependence on $\eta/\off$ in the preceding theorems is tight up to constant factors for some MTS instance.

For some specific MTS, however, the dependence on $\eta/\off$ can be
improved, as shown
in Section~\ref{sec:paging}, where we present a new algorithm for caching
whose competitive ratio has a logarithmic dependence on $\eta/\off$.
One of the main characteristics of our algorithm, which we call \algname, compared to previous approaches, is that it is able to \emph{gradually} adapt its level of trust in the predictor throughout the instance. {Showing that our general prediction setup can be used to design such efficient algorithms for caching is the most involved result of our paper, so the following result is proved before Theorems~\ref{thm:ub-det} and \ref{thm:ub-rand}.}
\begin{theorem}\label{thm:paging}
For caching with action predictions, there is a randomized algorithm with competitive ratio
$\min\big\{f(\frac\eta\off), g(k)\big\}$ against any algorithm $\off$,
where $f(\frac\eta\off) \leq O(\log \frac\eta\off)$ is
the smoothness with prediction error $\eta$ and
$g(k) \leq O(\log(k))$ is the robustness
with cache size $k$.
\end{theorem}
{We do not attempt to optimize constant factors in the proof of Theorem~\ref{thm:paging}, but we remark that $f$ can be chosen such that $f(0)= 1+\epsilon$, for arbitrary $\epsilon>0$. The reason is that our algorithm in the proof of Theorem~\ref{thm:paging} can be used as algorithm $A$ in Theorem~\ref{thm:ub-rand}.}

Although we designed our prediction setup with MTS in mind, it can 
also be applied to problems beyond MTS.
We demonstrate this in Section~\ref{sec:matching} by employing our techniques to provide an algorithm of similar flavor for
\emph{online matching on the line}, a problem not known to be an MTS.
\begin{theorem}\label{thm:matching}
For online matching on the line with action predictions, there is a deterministic algorithm with competitive ratio
$O(\min\{\log n, 1 + \eta/\off\})$, where $\eta$ is the prediction error with respect to some
offline algorithm \off.
\end{theorem}
We also show that Theorem~\ref{thm:matching} can be generalized to
give a $O(\min\{2n-1, \eta/\off\})$-competitive algorithm for \emph{online
metric bipartite matching}.

In Section~\ref{sec:limitations}, we show that the reoccurrence time predictions
introduced by \citet{LykourisV18} for caching
do not help for more general MTS.

\begin{theorem}\label{thm:limitations}
The competitive ratio of any algorithm for weighted caching
even if provided with precise reoccurrence time predictions
is $\Omega(\log k)$.
\end{theorem}
Note that there are $O(\log k)$-competitive online algorithms for
weighted caching which do
not use any predictions~\citep[see][]{BansalBN12}. This motivates the need for a different prediction setup as introduced in this paper.
This lower bound result has been obtained independently by \citet{JiangPS20} who also
proved a lower bound of $\Omega(k)$ for deterministic algorithms with precise
reoccurrence time predictions.
However, for instances with only $\ell$ weight classes,
\citet{BansalCKPV20} showed that perfect reoccurrence time predictions allow
achieving a competitive ratio of $\Theta(\log \ell)$.

We round up by presenting an extensive experimental evaluation
of our results that suggests practicality. We test the performance of
our algorithms on public data with previously used models. With respect to
caching, our algorithms outperform all previous approaches in most
settings (and are always at least comparable). A very interesting use of our
setup is that it allows us to employ any other online algorithm as a predictor for our algorithm. For instance, when using the Least Recently Used (LRU) algorithm -- which is considered the gold standard in practice -- as a predictor for our algorithm, our experiments suggest that we achieve the same practical performance as LRU, but with an exponential improvement in the theoretical worst-case guarantee
($O(\log k)$ instead of $k$). Finally we
applied our general algorithms to a simple MTS called the ice cream
problem and were able to obtain results that also suggest practicality
of our setup beyond caching.

\subsection{Related work}
Our work is part of a larger and recent movement to prove
rigorous performance guarantees for algorithms based on machine
learning. 
The first main results have been established on both
classical \citep[see][]{KraskaBCDP18, KhalilDNAS17}
and online problems: \citet{LykourisV18} and \citet{Rohatgi20} on caching,
\citet{LattanziLMV20} on restricted assignment scheduling,
\citet{PurohitSK18} on ski rental and non-clairvoyant scheduling,
\citet{GollapudiP19} on ski rental with multiple predictors,
\citet{Mitzenmacher20} on scheduling/queuing, and
\citet{MedinaV17} on revenue optimization.

Most of the online results are analyzed by means of
\emph{consistency} (competitive ratio in the case of perfect
predictions) and \emph{robustness} (worst-case competitive-ratio
regardless of prediction quality), which was first defined in this
context by \citet{PurohitSK18}, while \citet{Mitzenmacher20} uses a different
measure called \emph{price of misprediction}.
It should be noted that the exact definitions of consistency and robustness are
slightly inconsistent between different works in the literature, making it
often difficult to directly compare results.

\paragraph{Results on Caching.}
Among the closest results to our work are the ones by \citet{LykourisV18} and \citet{Rohatgi20}, who study the caching problem (a special case of MTS) with machine learned predictions.
\citet{LykourisV18} introduced the following prediction setup for caching:
whenever a page is requested, the algorithm receives a prediction of the time when the same page will be requested again. The prediction error is defined as the $\ell_1$-distance between the predictions and the truth, i.e., the sum -- over all requests -- of the
absolute difference between the predicted and the real reoccurrence time of the same request.
For this prediction setup, they adapted the classic Marker algorithm in order
to achieve, up to constant factors, the best robustness and consistency
possible.
In particular, they achieved a competitive ratio of
$O\big(1+\min \{ \sqrt{\eta/\opt},\log k \}\big)$ and their
algorithm was shown to perform well in experiments.
Later, \citet{Rohatgi20} achieved a better dependency on the prediction error:
$O\big(1+\min \big\{ \frac{\log k}{k} \frac{\eta}{\opt},\log k \big\}\big)$.
He also provides a close lower bound.

Following the original announcement of our work, we learned about further
developments by \citet{AlexWei20} and \citet{JiangPS20}.
\citet{AlexWei20}
{further refined the aforementioned results for caching with reoccurrence time predictions.}
The paper by \citet{JiangPS20} proposes an algorithm for weighted caching
in a very strong prediction setup, where
the predictor reports at each time step the reoccurrence
time of the currently requested page as well as
\emph{all} page requests up to that time. \citet{JiangPS20} provide a collection of
lower bounds for weaker predictors
(including an independent proof of Theorem~\ref{thm:limitations}),
justifying the need for such a strong predictor.
In a followup work, \citet{BansalCKPV20} showed, though, that the reoccurrence time
predictions\footnote{{Their actual algorithm only needs the relative ordering of reoccurrence times, which is also true for~\cite{LykourisV18,Rohatgi20,AlexWei20}.}} are still useful for weighted caching when the number $\ell$ of weight classes is small,
allowing to achieve a competitive ratio of $\Theta(\log \ell)$ for good predictions.

We stress that the aforementioned results use different prediction setups and
they do not imply any bounds for our setup. This is due to a different
way of measuring prediction errors, see Section~\ref{sec:setup_prop} for details. Therefore, we cannot compare the theoretical guarantees achieved by
previously published caching algorithms in their prediction setup to our new caching
algorithm within our broader setup.
Instead, we provide a comparison via experiments.\footnote{One might be tempted to adapt the algorithm of \citet{Rohatgi20}
to action
predictions by replacing the page with the furthest predicted reoccurrence
in the algorithm of \citet{Rohatgi20} by a page evicted by the predictor in our setting.
However, it is not hard, following ideas similar to the first example about prediction errors in the Section~\ref{sec:setup_prop}, to construct an instance where this algorithm is $\Omega(\log k)$-competitive although $\frac{\eta}{\opt}=O(1)$ in our setup.}

\paragraph{Combining Worst-Case and Optimistic Algorithms.}
An approach in some ways similar to ours was developed by
\citet{MahdianNS12}, who assume the existence of an optimistic
algorithm and developed a meta-algorithm that combines this algorithm
with a classical one and obtains a competitive ratio that is an
interpolation between the ratios of the two algorithms.  They designed
such algorithms for several problems including facility location and
load balancing. The competitive ratios obtained depend on the
performance of the optimistic algorithm and the choice of the
interpolation parameter. Furthermore the meta-algorithm is designed on
a problem-by-problem basis. In contrast, (i) our performance guarantees
are a function of the prediction error, (ii) generally we are able to
approach the performance of the best algorithm, and (iii) our way of
simulating multiple algorithms can be seen as a black box and is
problem independent.  

\paragraph{Online Algorithms with Advice.}
Another model for augmenting online algorithms, but not directly related to
the prediction setting studied in this paper, is that of \emph{advice
  complexity}, where information about the future is
obtained in the form of some always correct bits of advice (see \cite{BoyarFKLM17} for a survey). 
\citet{EmekFKR11} considered MTS under advice complexity,
and \citet{Angelopoulos19} consider advice complexity with possibly
adversarial advice and focus on Pareto-optimal algorithms for
consistency and robustness in several similar online problems.

\subsection{Comparison to the setup of Lykouris and Vassilvitskii}
\label{sec:setup_prop}

Although the work of \citet{LykourisV18} for caching served as an inspiration,
our prediction setup cannot be understood as an extension or generalization
of their setup. Here we list the most important connections and differences.

\paragraph{Conversion of Predictions for Caching.}
One can convert the reoccurrence time predictions of \citet{LykourisV18}
for caching into
predictions for our setup using a natural algorithm:
At each page fault, evict the
page whose next request is predicted furthest in the future.
Note that, if given perfect predictions, this algorithm produces an optimal
solution \citep{Belady66}.
The states of this algorithm at each time are then interpreted as predictions
in our setup.
We use this conversion to compare the performance of our
algorithms to those of \citet{LykourisV18} and \citet{Rohatgi20} in empirical
experiments in Section~\ref{sec:experiments}.

\paragraph{Prediction Error.}
The prediction error as defined by \citet{LykourisV18} is not directly
comparable to ours.
Here are two examples.

(1) Consider a paging instance where  some page $p$ is requested at times $1$ and $3$, and suppose we are given reoccurrence time predictions that are almost perfect except at time $1$ where it is predicted that $p$ reoccurs at time $T$ rather than $3$, for some large $T$. Then the prediction error in the setting of \citet{LykourisV18} is $\Omega(T)$. However, the corresponding action predictions obtained by the conversion above are wrong only at time step $2$, meaning the prediction error in our setting is only $1$ with respect to the offline optimum.

(2) One can create a request sequence consisting of $k+1$ distinct pages
where swapping two predicted times of next arrivals causes a different
prediction to be generated by the conversion algorithm.
The modified prediction in the setup of \citet{LykourisV18} may only have error $2$ while the error in our setup
with respect to the offline optimum can be arbitrarily high (depending on how
far in the future these arrivals happen).
However, our results provide meaningful bounds also in this situation.
Such predictions still have error $0$ in our setup with respect
to a near-optimal algorithm which incurs only one additional page fault
compared to the offline optimum.
Theorems
\ref{thm:ub-det}--\ref{thm:paging} then
provide constant-competitive algorithms with respect to this near-optimal
algorithm.

The first example shows that the results of
\citet{LykourisV18,Rohatgi20,AlexWei20} do not imply any bounds in our setup.
On the other hand, the recent result of \citet{AlexWei20} shows that our
algorithms from Theorems \ref{thm:ub-det}--\ref{thm:paging}, combined with the
prediction-converting algorithm above, are
$O(1+\min\{\frac1k\frac{\eta}\opt, \log k\})$-competitive
for caching in the setup of \citet{LykourisV18}, thus also matching the best known competitive ratio in that setup:
The output of the conversion algorithm has error 0 with respect
to itself and our algorithms are constant-competitive with respect to it.
Since the competitive ratio of the conversion algorithm is
$O(1+\frac1k \frac{\eta}{\opt})$ by \citet{AlexWei20},
our algorithms are
$O(\min\{1+\frac1k \frac{\eta}{\opt}, \log k\})$-competitive,
where $\eta$ denotes the prediction error in the setup of \citet{LykourisV18}.

\paragraph{Succinctness.}
In the case of caching, we can restrict ourselves to
\emph{lazy} predictors, where each predicted cache content differs from the previous predicted cache content by at most one page, and only if the previous predicted cache content did not contain the requested page. This is motivated by the fact that any algorithm can be transformed into a lazy version of itself without increasing its cost.
Therefore, $O(\log k)$ bits are enough to describe each action prediction,
saying which page should be evicted,
compared to $\Theta(\log T)$ bits needed to encode a reoccurrence time
in the setup of \citet{LykourisV18}. In fact, we need to receive a prediction not for all time steps but only those when the current request is not part of the previous cache content of the predictor.
In cases when running an ML predictor at each of these time steps is too costly,
our setup allows predictions being generated by some fast heuristic
whose parameters can be recalculated by the ML algorithm only when needed.

\paragraph{Learnability.}
In order to generate the reoccurrence time predictions,
\citet{LykourisV18} used a simple PLECO \citep{PLECO} predictor.
In this paper, we introduce another simple predictor called POPU
and show that the output of these predictors can be converted to
action predictions.

Predictors \emph{Hawkey} \cite{Hawkey16} and \emph{Glider} \cite{Glider19}
use binary classifiers to
identify pages in the cache which are going to be reused soon,
evicting first the other ones.
As shown by their empirical results, such binary information is enough to
produce a very efficient cache replacement policy, i.e.,
action predictions.
In their recent paper, \citet{LiuHSRA20} have proposed a new predictor, called
\emph{Parrot}, that is trained using the imitation learning approach and tries
to mimic the behaviour of the optimal offline algorithm~\citep{Belady66}.
The main output of their model, implemented using a neural network, are
in fact action predictions.
However it
also produces the reoccurrence time predictions in order
to add further supervision during the training process.
While at first it may seem that predicting reoccurrence times is an easier
task (in particular, it has the form of a standard supervised learning task),
the results of \citet{LiuHSRA20} show that it might well be the opposite -- e.g.,
when the input instance variance makes it impossible to predict the 
reoccurrence times accurately yet it is still possible to solve it (nearly) optimally online.
We refer to the paper of \citet{Adametal21} 
for an extensive evaluation of the existing learning augmented algorithms
using both reoccurrence time and action predictions.
{Following the emergence of learning-augmented algorithms, Anand et al.~\cite{anand2020customizing} even designed predictors specifically tuned to optimize the error used in the algorithms analysis. This work has been restricted so far to a much simpler online problem, ski rental.}

\section{Preliminaries}
\label{sec:prelim}

In MTS, we are given a metric space
$M$ of states and an initial state
$x_0 \in M$. At each time $t = 1, 2, \dotsc$, we receive
a task $\ell_t \colon M\to \R^+ \cup \{0,+\infty\}$
and we have to choose a new state $x_t$ without knowledge of the
future tasks, incurring cost
$\dist(x_{t-1},x_t) + \ell_t(x_t)$.
Note that $\dist(x_{t-1},x_t)=0$ if $x_{t-1} = x_t$
by the identity property of metrics.

Although MTS share several similarities with the {\em experts} problem
from the theory of online learning \citep{FreundS97,Chung94},
there are three important differences.
First, there is a {\em switching cost}:
we need to pay cost for switching between states
equal to their distance in the underlying metric space.
Second, an algorithm for MTS has {\em one-step lookahead}, i.e., it can see
the task (or loss function) before choosing the new state and
incurring the cost of this task.
Third, there can be {\em unbounded costs} in MTS, which can be handled thanks to
the lookahead.
See \citet{BlumB00} for more details on the relation between
experts and MTS.

To assess the performance of algorithms, we use the
{\em competitive ratio} --
the classical measure used in online algorithms.

\begin{definition}[Competitive ratio]
Let $\alg$ be an online algorithm for some cost-minimization problem $P$.
We say that $\alg$ is $r$-{\em competitive} and call $r$ the
{\em competitive ratio} of $\alg$,
if for any input sequence $I \in P$, we have
\[ \E[\cost(\alg(I))] \leq r\cdot \opt_I + \alpha, \]
where $\alpha$ is a constant independent of the input sequence,
$\alg(I)$ is the solution produced by the online algorithm and
$\opt_I$ is the cost of an optimal solution computed offline with the prior
knowledge of the whole input sequence.
The expectation is over the randomness in the online algorithm. If $\opt_I$ is replaced by the cost of some specific algorithm $\off$, we say that $\alg$ is $r$-competitive against $\off$.
\end{definition}

{Before we prove our results for general MTS, we consider in the next section the caching problem. It corresponds to the special case of MTS where the metric space is the set of cardinality-$k$ subsets of a cardinality-$n$ set (of pages), the distance between two sets is the number of pages in which they differ, and each cost function assigns value $0$ to all sets containing some page $r_t$ and $\infty$ to other sets.}

\section{Logarithmic Error Dependence for Caching}
\label{sec:paging}

We describe in this section
a new algorithm, which we call
\algname, for the (unweighted) caching problem, and prove Theorem~\ref{thm:paging}. The algorithm
achieves a competitive ratio logarithmic in the error
(thus overcoming the lower bound of Theorem~\ref{thm:generalLB} that holds for general MTS even on a uniform metric),
while also attaining the optimal worst-case guarantee of $O(\log k)$.

We assume that the predictor is \emph{lazy} in the following sense.
Let $r_t$ be the page that is requested at time $t$ and let $P_t$ be the
configuration (i.e., set of pages in the cache) of the predictor at time $t$.
Then $P_t$ differs from $P_{t-1}$ only if $r_t\notin P_{t-1}$ and, in this case, $P_t= P_{t-1}\cup\{r_t\}\setminus \{q\}$ for some page $q\in P_{t-1}$.
Note that any algorithm for caching can be converted into a lazy one without increasing its cost.

We partition the request sequence into \emph{phases}, which are maximal time periods where $k$ distinct pages are requested\footnote{Subdividing the input sequence into such phases is a very common technique in the analysis of caching algorithms, see for example~\citet{BorodinEY1998} and references therein.}: The first phase begins with the first request. A phase ends (and a new phase begins) after $k$ distinct pages have been requested in the current phase and right before the next arrival of a page that is different from all these $k$ pages. For a given point in time, we say that a page is \emph{marked} if it has been requested at least once in the current phase. For each page $p$ requested in a
phase, we call the first request to $p$ in that phase the
\emph{arrival} of $p$. This is the time when $p$ gets marked. Many algorithms, including that of \citet{LykourisV18}, belong to the class of so-called \emph{marking algorithms}, which evict a page only if it is unmarked. The classical $O(\log k)$-competitive online algorithm of \citep{FiatKLMSY91} is a particularly simple marking algorithm: On a cache miss, evict a uniformly random unmarked page. In general, no marking algorithm can be better than $2$-competitive even when provided with perfect predictions.
As
will become clear from the definition of \algname later, it may follow the predictor's advice to evict even marked pages, meaning that it is not a marking algorithm.
As can be seen in our experiments in Section~\ref{sec:experiments}, this allows \algname to outperform previous algorithms when predictions are good.\footnote{There exist instances where \algname with perfect predictions strictly outperforms the best marking algorithm, but also vice versa, see Appendix~\ref{sec:TDvsMarking}.} We believe that one could modify the algorithm so that it is truly $1$-competitive in the case of perfect predictions. However, formally proving so seems to require a significant amount of additional technical complications regarding notation and algorithm description. To keep the presentation relatively simple, we abstain from optimizing constants here.%

\subsection{First warm-up: A universe of $\boldsymbol{k+1}$ pages}
Before we give the full-fledged \algname algorithm for the general setting, we first describe an algorithm for the simpler setting when there exist only $k+1$ different pages that can be requested. This assumption substantially simplifies both the description and the analysis of the algorithm while already showcasing some key ideas. In Sections~\ref{sec:MarkingPredictor} and~\ref{sec:TnDgeneral}, we will explain the additional ideas required to extend the algorithm to the general case.

Our assumption means that at each time, there is only one page missing from the algorithm's cache and only one page missing from the predicted cache. Moreover, the first request in each phase is an arrival of the (unique) page that was not requested in the previous phase, and all other arrivals in a phase are requests to pages that were also requested in the previous phase.

\subsubsection{Algorithm (simplified setting)}

We denote by $M$ the set of marked pages and by $U$ the set of unmarked pages.

In each phase, we partition time into alternating \emph{Trust intervals} and \emph{Doubt intervals}, as follows: When a phase starts, the first Trust interval begins. Throughout each Trust interval, we ensure that the algorithm's cache is equal to the predicted cache $P_t$. As soon as the page missing from $P_t$ is requested during a Trust interval, we terminate the current Trust interval and start a new Doubt interval. In a Doubt interval, we treat page faults by evicting a uniformly random page from $U$. As soon as there have been $2^{i-1}$ arrivals since the beginning of the $i$th Doubt interval of a phase, the Doubt interval ends and a new Trust interval begins (and we again ensure that the algorithm's cache is equal to $P_t$).

\subsubsection{Analysis (simplified setting)}
Let $d_\ell$ be the number of Doubt intervals in phase $\ell$.

\begin{claim}\label{cl:costdl}
The expected number of cache misses in phase $\ell$ is $1+O(d_\ell)$.
\end{claim}
\begin{proof}
Any cache miss during a Trust interval starts a new Doubt interval, so there are $d_\ell$ cache misses during Trust intervals. There may be one more cache miss at the start of the phase. It remains to show that there are $O(d_\ell)$ cache misses in expectation during Doubt intervals.

In a Doubt interval, we can have a cache miss only when a page from $U$ arrives. The arriving page from $U$ is the one missing from the cache with probability $1/|U|$. Moreover, when a page from $U$ arrives, it is removed from $U$. The expected number of cache misses during Doubt intervals is therefore a sum of terms of the form $1/|U|$ for distinct values of $|U|$. Since the total number of arrivals during Doubt intervals is at most $2^{d_\ell}$, the expected number of cache misses during Doubt intervals is at most
$\sum_{u=1}^{2^{d_\ell}} 1/u = O(d_\ell)$.%
\end{proof}

Due to the claim, our main remaining task is to upper bound the number of Doubt intervals.

We call a Doubt interval \emph{error interval} if at each time $t$ during the interval, the page missing from $P_t$ is present in the cache of the offline algorithm. Note that each time step during an error interval contributes to the error $\eta$. Let $e_\ell\le d_\ell$ be the number of error intervals of phase $\ell$. Since for all $i<e_\ell$, the $i$th error interval contains at least $2^{i-1}$ time steps, we can bound the error as
\begin{align}
\eta\ge\sum_\ell \left(2^{e_\ell-1}-1\right).\label{eq:errorSimple}
\end{align}
Denote by $\off_\ell$ the cost of the offline algorithm during phase $\ell$.

\begin{claim}\label{cl:boundOnd}
	$d_\ell\le\off_\ell+e_\ell+1$.
\end{claim}
\begin{proof}
Consider the quantity $d_\ell-e_\ell$. This is the number of Doubt intervals of phase $\ell$ during which $P_t$ is equal to the offline cache at some point. Since $P_t$ changes at the start of each Doubt interval, but $P_t$ changes only if the page missing from $P_t$ is requested (since we assume the predictor to be lazy), this means that the offline cache must change between any two such intervals. Thus, $d_\ell-e_\ell-1\le \off_\ell$.
\end{proof}

Combining these claims, the total number of cache misses of the algorithm is at most (noting by $\#\textit{phases}$ the number of phases)
\begin{align*}
\sum_{\ell}\left(1+O(d_\ell)\right)\le O\left(\off + \#\textit{phases} + \sum_\ell (e_\ell-1)\right).
\end{align*}
Each term $(e_{\ell}-1)$ can be rewritten as $\log_2\big(1+(2^{e_{\ell}-1}-1)\big)$. By concavity of $x\mapsto\log(1+x)$, subject to the bound \eqref{eq:errorSimple} the sum of these terms is maximized when each term $2^{e_{\ell}-1}-1$ equals $\frac{\eta}{\#\textit{phases}}$. Thus, the total number of cache misses of the algorithm is at most
\begin{align*}
O\left(\off + \#\textit{phases} + \#\textit{phases}\cdot \log\left(1+\frac{\eta}{\#\textit{phases}}\right)\right)\le \off \cdot O\left(1 + \log\left(1 + \frac{\eta}{\off}\right)\right),
\end{align*}
where the last inequality uses that $\off =\Omega(\#\textit{phases})$ since all $k+1$ pages are requested in any two adjacent phases, so the offline algorithm must have a cache miss in any two adjacent phases.

\subsection{Second warm-up: The predictor is a marking algorithm}\label{sec:MarkingPredictor}
We now drop the assumption from the previous section and allow the number of pages in the universe to be arbitrary. However, we will assume in this section that the predictor is a marking algorithm (i.e., the predicted cache $P_t$ always contains all marked pages). In this case, our algorithm will also be a marking algorithm. %

Our algorithm is again based on phases, which are defined as before. Denote by $U$ the set of unmarked pages that were in the cache at the beginning of the phase, and by $M$ the set of marked pages. By our assumption that both the predictor and our algorithm are marking algorithms, at the start of a phase $U$ is equal to both the predicted cache as well as the algorithm's cache as it contains precisely the pages that were requested in the previous phase. An important notion in phase-based paging algorithms is that of \emph{clean pages}. For the setting considered in this section, where the predictor is a marking algorithm, we define a page as \emph{clean} if it is requested in the current phase but was not requested in the previous phase. We denote by $C$ the set of clean pages that have arrived so far in the current phase.
(In the general setting, we will need to define clean pages slightly differently.)

\paragraph{Several simultaneous interval partitions.} While in the first warm-up setting with a $(k+1)$-page universe there could be only a single clean page per phase, a main difference now is that there can be \emph{several} clean pages in a phase. For this reason, it is no longer sufficient to partition the phase into Trust intervals and Doubt intervals that are defined ``globally''. Instead, we will associate with \emph{each} clean page a \emph{different} subdivision of time into Trust intervals and Doubt intervals: The time from the arrival of $q$ until the end of the phase is partitioned into alternating \emph{$q$-Trust intervals} and \emph{$q$-Doubt intervals}. Thus, a time $t$ can belong to various intervals -- one associated to each clean $q$ that has arrived so far in the current phase. At any time, some of the current intervals may be Trust intervals while the rest are Doubt intervals. During a $q$-Trust interval, we will \emph{not} ensure that the entire algorithm's cache is equal to the predictor's cache, but only that one particular page $f(q)$ that is evicted by the predictor is also predicted by the algorithm.

More precisely, we will also maintain a map $f\colon C\to U\setminus P_t$ that maps each clean page $q\in C_\ell$ (that has arrived so far) to an associated page $f(q)$ that was evicted by the predictor during the current phase (and is currently still missing from the predictor's cache). Intuitively, we can think of $f(q)$ as the page that the predictor advises us to evict to make space for $q$. If it happens that the page $f(q)$ associated to some clean $q$ is requested, the predictor has to load $f(q)$ back to its cache $P_t$, and we redefine $f(q)$ to be the page that the predictor evicts at this time. %
Observe that this ensures that the pages $f(q)$ are distinct for different $q$ (in fact, since we assume the predictor to be a lazy marking algorithm, $f$ is a bijection in this case).

When a clean page $q$ arrives, the first $q$-Trust interval begins. Throughout each $q$-Trust interval, we will ensure that the page $f(q)$ is evicted from our algorithm's cache. If during a $q$-Trust interval the page $f(q)$ is requested, we terminate the current $q$-Trust interval and start a new $q$-Doubt interval. In a $q$-Doubt interval, we ignore the advice to evict $f(q)$ and instead evict a uniformly random unmarked page when necessary. As soon as there have been $2^{i-1}$ arrivals since the beginning of the $i$th $q$-Doubt interval, the $q$-Doubt interval ends and a new $q$-Trust interval begins (and we again ensure that the page currently defined as $f(q)$ is evicted).

We will skip a more formal description and analysis of the algorithm for this setting as it will be contained as a special case of our algorithm in the next section.

{\begin{remark}
		At a high level, the idea of linking evictions to individual clean pages (which is explicit for the pages $f(q)$ evicted in Trust intervals) bears some similarities to the notion of eviction chains used in~\cite{LykourisV18,Rohatgi20}. However, our algorithm and charging scheme are quite different. In particular, the natural adaptations of algorithms in~\cite{LykourisV18,Rohatgi20} to our setting would only be $\Omega(\log k)$-competitive even when $\frac{\eta}{\opt}=O(1)$, where $\eta$ is the prediction error in our setting. This can happen on instances where predictions are mostly good, but occasionally very bad. To overcome this, we use Doubt intervals that start small and grow over time, which allows our algorithm to recover quickly from occasional very bad predictions.
\end{remark}}

\subsection{Algorithm for the general case}
\label{sec:TnDgeneral}

We now describe our algorithm \algname for the general case. In contrast to the previous section, we drop here the assumption that the predictor must be a marking algorithm. Thus, the predictor may evict marked pages, and since \algname may trust such evictions, also \algname may evict marked pages. Consequently, it is no longer true that set of pages in the algorithm's cache at the start of a phase is equal to the set of pages requested in the previous phase. This means that some pages may be ``ancient'' as per the following definition. 

\begin{definition} A page is called
\emph{ancient} if it is in \algname's cache even though it has been requested in
neither the previous nor the current phase (so far).
\end{definition}

We partition each phase into two stages that are determined by whether ancient pages exist or not: During stage one there exists at least one ancient page, and during stage two there exist no ancient pages. We note that one of the two stages may be empty.

The algorithm for stage one is very simple: Whenever there is a page fault, evict an arbitrary ancient page. This makes sense since ancient pages have not been requested for a long time, so we treat them like a reserve of pages that are safe to
evict. Once this reserve has been used up, stage two begins.

The algorithm for stage two is essentially the one already described in the previous section. Before we give a more formal description, we first fix some notation. Let $U$ be the set of pages that were in cache at the beginning of stage two and that are currently unmarked. Let $M$ be the set of marked pages. We call a page \emph{clean} for a phase if it arrives in stage two and it was not in $U\cup M$ immediately before its arrival. (Pages arriving in stage one are \emph{not} considered clean as these are easy to charge for and do not need the analysis linked to clean pages in stage two.) By $C$ we denote the set of clean pages that have arrived so far in the current phase.

\begin{figure}[tbh]
  \centering
  \newlength{\scal}
\setlength{\scal}{0.6cm}
\newcommand{\padphase}{0.2}
	\begin{tikzpicture}[x=\scal, y=\scal, page/.style={rectangle, draw, black, minimum width=\scal, minimum height=\scal,font=\vphantom{b}}, 
		phase/.style={very thick, red }, arrive/.style={fill=yellow!50!white}]
		
		\node[page, arrive] (a) at (0,0) {a};
		\node[page, anchor = east]  at (-0.5,0) {\dots};
		\node at (-3,0) {$k=3$} ;
		\node[page, arrive] at (1,0) {b};
		\node[page] at (2,0) {a};
		\node[page, arrive] at (3,0) {c};
		\node[page] at (4,0) {a};
		\node[page, arrive] at (5,0) {e};
		\node[page, arrive] at (6,0) {b};
		\node[page, arrive] at (7,0) {a};
		\node[page] at (8,0) {e};
		\node[page, arrive] at (9,0) {c};
		\node[page, anchor=west] at (9.5,0) {\dots};
		
		\draw[phase] (-0.5,-\scal/2-\padphase*\scal) -- + ( 0,\scal+2*\padphase*\scal)
		(4.5,-\scal/2-\padphase*\scal) edge node [at end, above] {t}   + ( 0,\scal+2*\padphase*\scal)
		(8.5,-\scal/2-\padphase*\scal) -- + ( 0,\scal+2*\padphase*\scal);
		
		\node at (4.5,-2) {\algname{} possible cache at time $t$: \texttt{\{a,c,\colorbox{blue!30}{d}\}}};
		
		\node[arrive] at (-0.5,1.5) {{\it arrivals}};
		\node[red] at (10,1.5) {{\it end of phases}};
		
		\node[fill= blue!30] at (13,-1.5) {{\it ancient}};
		
	\end{tikzpicture}
	\caption{Illustration of definitions used to describe \algname. At time $t$, a new phase starts and the cache contains $a$, $c$ and $d$, where $d$ is ancient. In the following phase, $e$ is not clean because it arrives during stage one. When $e$ is requested, \algname evicts the ancient page $d$ and loads $e$, and then stage two begins with $U=\{a,c\}$ and $M=\{e\}$ initially. The page $b$ requested next is considered clean because, although it was also requested in the previous phase, it was not in $U\cup M$ immediately before its request. The next requested page $a$ is not clean as it was already in~$U\cup M$.}
	\label{fig:markingillu}
\end{figure}

It is immediate from the definitions that the following equation is maintained during stage two:
\begin{align}
|U\cup M|= k + |C|\label{eq:UMeqkC}
\end{align}

Similarly to before, \algname maintains an injective map $f\colon C\to (U\cup M)\setminus P_t$ that maps each clean page $q\in C$ (that has arrived so far) to a distinct page $f(q)$ that is currently missing from the predictor's cache. Note that since the predictor may evict marked pages, it is necessary to include marked pages in the codomain of $f$. As before, the time from the arrival of a clean page $q$ to the end of the phase is partitioned into alternating $q$-Trust intervals and $q$-Doubt intervals. Depending on the type of the current interval, we will also say that $q$ is trusted or $q$ is doubted. Let
\begin{align*}
	T&:=\{f(q)\mid q\in C\text{ trusted}\}\\
	D&:=\{f(q)\mid q\in C\text{ doubted}\}.
\end{align*}
To organize the random evictions that the algorithm makes, we sort the pages of $U$ in a uniformly random order at the beginning of stage two. We refer to the position of a page in this order as its \emph{rank}, and we will ensure that the randomly evicted pages are those with the lowest ranks.\footnote{Since randomly evicted pages may be reloaded even when they are not requested, maintaining such ranks leads to consistent random choices throughout a phase.}

A pseudocode of \algname when a page $r$ is requested in stage two is given in Algorithm~\ref{alg:TnD}.

If $r$ is clean and the request is an arrival (the condition in line~\ref{line:clean} is true), we first define its associated page $f(r)$ as an arbitrary\footnote{e.g., the least recently used} page from $(U\cup M)\setminus(P_t\cup T\cup D)$. We will justify later in Lemma~\ref{lem:welldef} that this set is non-empty. We then start an $r$-Trust interval. (Note that this adds $f(r)$ to the set $T$.) Since $r$ is then trusted, we ensure that $f(r)$ is evicted from the cache. If it was already evicted, then we instead evict the page in cache with the lowest rank. Either way, there is now a free cache slot that $r$ will be loaded to in line \ref{line:loadClean}. We also initialize a variable $t_r$ as $1$. For each clean page $q$, we will use this variable $t_q$ to determine the duration of the next $q$-Doubt interval.

Otherwise, we also ensure that $r$ is in cache, evicting the page in cache with the lowest rank if necessary (lines~\ref{line:miss}--\ref{line:loadNonClean}).

If $r$ is a page of the form $f(q)$, we redefine $f(q)$, and since the previous prediction to evict the old $f(q)$ was bad, we ensure that $q$ is now doubted (lines~\ref{line:fqReq}--\ref{line:startDoubt}), {\it i.e.}, we start a new $q$-Doubt interval if $q$ was trusted. %

Finally, in lines~\ref{line:foreach}--\ref{line:evictReload} we check for each clean page $q$ whether it should reach the end of its Doubt interval. If $q$ is in its $i$th Doubt interval, then this happens if the current arrival is the $2^{i-1}$th arrival after the start of the current $q$-Doubt interval. For each $q$ for which a $q$-Doubt interval ends, we start a new $q$-Trust interval and ensure that $f(q)$ is evicted from the cache. If $f(q)$ was not evicted yet, we reload the evicted page with the highest rank back to the cache so that the cache contains $k$ pages at all times.

\LinesNumbered
\begin{algorithm2e}
	\caption{When page $r$ is requested in phase $\ell$ and no ancient pages exist}\label{alg:TnD}
	\DontPrintSemicolon
	\If(\tcp*[f]{Arrival of a clean page}){$r\in C$ and this is the arrival of $r$\label{line:clean}}{
		Let $f(r)$ be an arbitrary page from $(U \cup M)\setminus (P_t\cup T\cup D)$\label{line:setFr}\;
		Start an $r$-Trust interval\label{line:initialTrust}\;
		\lIf{$f(r)$ is in cache}{evict $f(r)$}
		\lElse{evict the lowest ranked cached page from $U\setminus T$\label{line:evictRanked1}}
		Load $r$ to the cache\label{line:loadClean}\;
		$t_r:=1$\tcp*[f]{Duration of next $r$-Doubt interval (if it exists)}\label{line:tr1}\;
	}
	\ElseIf(\tcp*[f]{Page fault, but not arrival of clean page}){$r$ is not in cache\label{line:miss}}{
		{
			Evict the lowest ranked cached page from $U\setminus T$\label{line:evictRanked2}\;
			Load $r$ to the cache}\label{line:loadNonClean}
	}
	\If(\tcp*[f]{Advice to evict $f(q)$ was bad}){$r=f(q)$ for some $q\in C$\label{line:fqReq}}{
		Redefine $f(q)$ as an arbitrary page from $(U\cup M)\setminus (P_t\cup T\cup D)$\label{line:redefineFq}\;
		\If{$q$ is trusted}{
			End the $q$-Trust interval and start a $q$-Doubt interval\label{line:startDoubt}\;
		}
	}
		\ForEach(\tcp*[f]{Check for end of Doubt intervals}){$q\in C$ that is doubted\label{line:foreach}}{
			\If{the current request is the $t_q$th arrival since the start of this $q$-Doubt interval}{
				End the $q$-Doubt interval and start a $q$-Trust interval\;
				$t_q:= 2\cdot t_q$\;
				\If{$f(q)$ is in cache}{
					Evict $f(q)$ and load the highest ranked evicted page from $U\setminus T$ back to the cache\label{line:evictReload}}
			}
		}
\end{algorithm2e}

\begin{remark}
\label{rmk:lazy}
To simplify the analysis, the algorithm is defined non-lazily here in the
sense that it may load pages even when they are not requested (in line~\ref{line:evictReload}).
An implementation should only simulate this
non-lazy algorithm in the background and, whenever the actual algorithm
has a page fault, it evicts an arbitrary (e.g., the least recently used)
page that is present in its own cache but missing from the simulated
cache.
\end{remark}

\paragraph{Correctness.}
It is straightforward to check that the algorithm's cache is always a subset of $(U\cup M)\setminus T$, since any page added to $T$ is evicted. Moreover, it is always a superset of $M\setminus T$ because pages from $M$ are only evicted if they are in $T$.

The following two lemmas capture invariants that are maintained throughout the execution of the algorithm. In particular, they justify the the algorithm is well-defined.
\begin{lemma}\label{lem:welldef}
	The set $(U \cup M)\setminus (P_t\cup T\cup D)$ is non-empty before $f(r)$ or $f(q)$ is chosen from it in lines~\ref{line:setFr} and~\ref{line:redefineFq}.
\end{lemma}
\begin{proof}
	It suffices to show that $|U\cup M| > |P_t\cup T\cup D|$ right before the respective line is executed.
	
	Before line~\ref{line:setFr} is executed, it will be the case that $|C|=|T\cup D|+1$ (because $r\in C$, but $f(r)$ is not defined yet). Thus, equation \eqref{eq:UMeqkC} yields $|U\cup M|= k + |T\cup D|+1 > |P_t\cup T\cup D|$.
	
	Before line~\ref{line:redefineFq} is executed, it holds that $|C|=|T\cup D|$ and $r\in (T\cup D)\cap P_t$. Again, the inequality follows from equation \eqref{eq:UMeqkC}.
\end{proof}

The next lemma justifies that reloading a page in line~\ref{line:evictReload} will be possible, and the lemma will also be crucial for the competitive analysis later.
\begin{lemma}\label{lem:Dmissing}
	Before each request of stage two, there are $|D|$ pages from $U\setminus T$ missing from the cache.
\end{lemma}
\begin{proof}
	The pages in cache are a subset of $U\cup M$ of size $k$. By equation~\eqref{eq:UMeqkC}, there are $|C|=|T\cup D|$ of these pages missing from the cache. The pages in $T$ account for $|T|$ of those missing pages. The remaining $|D|$ missing pages are all in $U\setminus T$ (because pages from $M$ are only evicted if they are in~$T$).
\end{proof}

\subsection{Competitive analysis}
Let $C_\ell$ denote the set $C$ at the end of phase $\ell$. The next
lemma and its proof are similar to a statement in \citet{FiatKLMSY91}. However, since
our definition of clean pages is different, we need to reprove it in our setting.
\begin{lemma}\label{lem:cachingOpt}
	Any offline algorithm suffers cost at least
	\begin{align*}
	\off\ge\Omega\left(\sum_{\ell}|C_\ell|\right).
	\end{align*}
\end{lemma}
\begin{proof}
	We first claim that at least $k+|C_\ell|$ distinct pages are requested in phases $\ell-1$ and $\ell$ together. If there is no stage two in phase $\ell$, then $C_\ell$ is empty and the statement trivial. Otherwise, all pages that are in $U\cup M$ at the end of phase $\ell$ were requested in phase $\ell-1$ or $\ell$, and by equation~\eqref{eq:UMeqkC} this set contains $k+|C_\ell|$ pages.
	
	Thus, any algorithm must suffer at least cost $|C_\ell|$ during these two phases. Hence, $\off$ is lower bounded by the sum of $|C_\ell|$ over all even phases and, up to a constant, by the according sum over all odd phases. The lemma follows.
\end{proof}

By the following lemma, it suffices to bound the cost of \algname incurred during stage two.
\begin{lemma}\label{lem:chargeAncient}
	The cost during stage one of phase $\ell$ is at most the cost during stage two of phase $\ell-1$.
\end{lemma}
\begin{proof}
The cost during stage one of phase $\ell$ is at most the number of ancient pages at the beginning of phase $\ell$. This is at most the number of marked pages that were evicted in phase $\ell-1$. Since a marked page can be evicted only during stage two, the lemma follows.
\end{proof}

Let $d_{q,\ell}$ be the number of $q$-Doubt intervals in phase $\ell$. The next lemma is reminiscent of Claim~\ref{cl:costdl} from our first warm-up section.

\begin{lemma}\label{lem:costCldlq}
	The expected cost during stage two of phase $\ell$ is $O\left(|C_\ell|+\sum_{q\in C_\ell}d_{q,\ell}\right)$.
\end{lemma}
\begin{proof}
	The cost incurred in lines~\ref{line:clean}--\ref{line:tr1} is at most $O(|C_\ell|)$. In lines~\ref{line:miss}--\ref{line:loadNonClean}, the algorithm can incur cost only if the requested page was in $U\cup T$ before the request (because the request is not an arrival of a clean page, so it was in $U\cup M$, and if it was in $M\setminus T$ then it was in cache already). If the page was in $T$, then a new Doubt interval will start in line~\ref{line:startDoubt}, so the cost due to those pages is at most $\sum_{q\in C_\ell}d_{q,\ell}$. If the page was in $U\setminus T$, then by Lemma~\ref{lem:Dmissing} and the random choice of ranks it was missing from the cache with probability $\frac{|D|}{|U\setminus T|}$. To account for this cost, we charge $\frac{1}{|U\setminus T|}$ to each clean $q$ that is doubted at the time. Over the whole phase, the number of times we charge to each $q\in C_\ell$ in this way is at most the total number of arrivals during $q$-Doubt intervals, which is at most $2^{d_{q,\ell}}$. By equation~\eqref{eq:UMeqkC} and since $|C|=|T|+|D|$, we have $|U|=k+|T|+|D|-|M|$, so $|U\setminus T|\ge k+|D|-|M|\ge k+1-|M|$. The quantity $|M|$ increases by $1$ after each such request to a page in $|U\setminus T|$, so the value of $|U\setminus T|$ can be lower bounded by $1,2,3,\dots,2^{d_{q,\ell}}$ during the at most $2^{d_{q,\ell}}$ arrivals when $\frac{1}{|U\setminus T|}$ is charged to $q$. Hence, the total cost charged to $q$ is at most $O(d_{q,\ell})$. It follows that the overall cost incurred in lines is at most$O\left(\sum_{q\in C_\ell}d_{q,\ell}\right)$
	
	Finally, the only other time cost is incurred is in line~\ref{line:evictReload}. This also amounts to at most $\sum_{q\in C_\ell}d_{q,\ell}$ because it happens only at the end of a Doubt-interval.
\end{proof}

Denote by $e_{q,\ell}$ the number of $q$-Doubt intervals with the property that at each time during the interval, the page currently defined as $f(q)$ is present in the offline cache. Since the current page $f(q)$ is never in the predictor's cache, and the $i$th doubted $q$-interval contains $2^{i-1}$ arrivals for $i<d_q$, a lower bound on the prediction error is given by
\begin{align}
	\eta\ge \sum_\ell\sum_{q\in C_\ell} (2^{e_{q,\ell}-1}-1).\label{eq:cachingError}
\end{align}

Denote by $\off_{q,\ell}$ the number of times in phase $\ell$ when the offline algorithm incurs cost for loading the page currently defined as $f(q)$ to its cache.

The next lemma is the generalization of Claim~\ref{cl:boundOnd}.
\begin{lemma}\label{lem:boundOndql}
	For each $q\in C_\ell$, we have $d_{q,\ell}\le\off_{q,\ell} + e_{q,\ell} +1$.
\end{lemma}
\begin{proof}
	Consider the quantity $d_{q,\ell}-e_{q,\ell}$. This is the number of $q$-Doubt intervals of phase $\ell$ during which $f(q)$ is missing from the offline cache at some point. Except for the last such interval, the page $f(q)$ will subsequently be requested during the phase, so the offline algorithm will incur cost for loading it to its cache. The lemma follows, with the ``$+1$'' term accounting for the last interval.
\end{proof}

We are now ready to prove the main result of this section.
\begin{theorem*}[Restated Theorem~\ref{thm:paging}]
	\algname has competitive ratio $O(\min\{1+\log(1+\frac{\eta}{\off}),\log k\})$ against any offline algorithm $\off$, where $\eta$ is the prediction error with respect to $\off$.
\end{theorem*}
\begin{proof}
	The $O(\log k)$ bound follows from Lemma~\ref{lem:costCldlq}, Lemma~\ref{lem:cachingOpt} and the fact that $d_{q,\ell}\le O(\log k)$ for each $q\in C_\ell$. The latter fact holds because if $d_{q,\ell}\ge 2$, then the $(d_{q,\ell}-1)$st $q$-Doubt interval contains $2^{d_{q,\ell}-2}$ arrivals, but there are only $k$ arrivals per phase.
	
	For the main bound, combining Lemmas~\ref{lem:chargeAncient}, \ref{lem:costCldlq} and \ref{lem:boundOndql} we see that the total cost of the algorithm is at most
	\begin{align*}
	O\left(\off+\sum_\ell|C_\ell|+\sum_\ell\sum_{q\in C_\ell}e_{q,\ell}\right).
	\end{align*}
	The summands $e_{q,\ell}$ can be rewritten as $1+\log_2\left(1+[2^{e_{q,\ell}-1}-1]\right)$.  By concavity of $x\mapsto\log(1+x)$, while respecting the bound \eqref{eq:cachingError} the sum of these terms is maximized when each term in brackets equals $\frac{\eta}{\sum_\ell|C_\ell|}$, giving a bound on the cost of
	\begin{align*}
	O\left(\off+\sum_\ell|C_\ell|\left(1+\log\left(1+\frac{\eta}{\sum_\ell|C_\ell|}\right)\right)\right).
	\end{align*}
	Since this quantity is increasing in $\sum_\ell|C_\ell|$, applying Lemma~\ref{lem:cachingOpt} completes the proof of the  theorem.
\end{proof}

\subsection{Lower bound}
The $O(\log k)$ upper bound matches the known lower bound $\Omega(\log k)$ on the competitive ratio of randomized online algorithms without prediction \cite{FiatKLMSY91}. %
 The competitive ratio of \algname when expressed only as a function of the error, $O(1+\log(1+\frac{\eta}{\off}))$, is also tight due to the following theorem. It should be noted, though, that for the competitive ratio as a function of both $k$ and $\frac{\eta}{\off}$ it is still plausible that a better bound can be achieved when $\frac{\eta}{\off}$ is relatively small compared to $k$.
\begin{theorem}
	If an online caching algorithm achieves competitive ratio at most $f(\frac{\eta}{\opt})$ for arbitrary $k$ when provided with action predictions with error at most $\eta$ with respect to an optimal offline algorithm $\opt$, then $f(x)=\Omega(\log x)$ as $x\to\infty$.
\end{theorem}
\begin{proof}
	Fix some $k+1$ pages and consider the request sequence where each request is to a uniformly randomly chosen page from this set. We define phases in the same way as in the description of \algname. By a standard coupon collector argument, each phase lasts $\Theta(k\log k)$ requests in expectation. An optimal offline algorithm can suffer only one page fault per page by evicting only the one page that is not requested in each phase. On the other hand, since requests are chosen uniformly at random, any online algorithm suffers a page fault with probability $1/(k+1)$ per request, giving a cost of $\Theta(\log k)$ per phase. Since $\frac{\eta}{\opt}=O(k\log k)$ due to the duration of phases, the competitive ratio of the algorithm is $\Omega(\log k)=\Omega(\log \frac{\eta}{\opt})$.
\end{proof}

\section{Robust Algorithms for MTS}
\label{sec:mts}

The goal of this section is to prove Theorem \ref{thm:ub-det} and Theorem
\ref{thm:ub-rand}, {which deal with algorithms substantially simpler than
\algname, but demonstrate the usefulness of our prediction setup for the broad
class of MTS problems. In Section~\ref{sec:ftp} we will first describe a simple algorithm whose competitive ratio depends linearly on the prediction error, but the algorithm is not robust against large errors. In Section~\ref{sec:combining} we then robustify this algorithm based on powerful methods from the literature. Finally, in Section~\ref{sec:lowerbounds} we show that the linear
dependency  of the achieved competitive ratio on $\eta/\opt$ is inevitable for
some MTS}.

\subsection{A non-robust algorithm}\label{sec:ftp}

We consider a simple memoryless algorithm, which we call \ftp.

\paragraph{Algorithm Follow the Prediction (\ftp).}
Intuitively, our algorithm follows the
predictions, but still somewhat cautiously: if there exists a state
``close'' to the predicted one that has a much cheaper service cost,
then it is to be preferred.
Let us consider a metrical task system with a set of states $X$.
We define the algorithm
\ftp (Follow the Prediction) as follows: at time $t$, after receiving
task $\ell_t$ and prediction $p_t$, it moves to the state
\begin{equation}\label{eq:ftp-def}
	x_t \gets \arg\min_{x\in X} \{\ell_t(x) + 2 \dist(x,p_t)\}.
\end{equation}
In other words,
\ftp follows the predictions except when it is beneficial to move from
the predicted state to some other state, pay the service and move back
to the predicted state.

\begin{lemma}\label{lem:ftp}
For any MTS with action predictions, algorithm $\ftp$ which achieves competitive ratio $1 + \frac{4\eta}{\off}$
against any offline algorithm \off, where $\eta$ is the prediction
error with respect to \off.
\end{lemma}
\begin{proof}%
At each time $t$, the \ftp algorithm is located at configuration $x_{t-1}$ and
needs to choose $x_t$ after receiving task $\ell_t$ and prediction $p_t$.
Let us consider some offline algorithm \off.
We denote $x_0, o_1, \dotsc, o_T$ the states of \off, where
the initial state $x_0$ is common for \off and for \ftp,
and $T$ denotes the length of the sequence.

We define $A_t$ to be the algorithm which agrees with $\ftp$ in its
first $t$ configurations $x_0, x_1, \dotsc, x_t$ and then agrees with the
states of \off, i.e., $o_{t+1}, \dotsc, o_T$.
Note that $\cost(A_0) = \off$ and $\cost(A_T) = \cost(\ftp)$.
We claim that $\cost(A_t)\leq \cost(A_{t-1}) + 4\eta_t$ for each $t$, where
$\eta_t = \dist(p_t, o_t)$.
The algorithms $A_t$ and $A_{t-1}$ are in the same configuration at each time
except $t$, when $A_t$ is in $x_t$ while $A_{t-1}$ is in $o_t$.
By the triangle inequality, we have
\begin{align*}
\cost(A_t)
	&\leq \cost(A_{t-1}) + 2\dist(o_t,x_t) + \ell_t(x_t) - \ell_t(o_t)\\
	&\leq \cost(A_{t-1}) + 2\dist(o_t,p_t) - \ell_t(o_t)
		+\; 2\dist(p_t,x_t) + \ell_t(x_t)\\
	&\leq \cost(A_{t-1}) + 4\dist(o_t,p_t),
\end{align*}
The last inequality follows from \eqref{eq:ftp-def}: we have
$2\dist(p_t,x_t)+\ell_t(x_t) \leq 2\dist(p_t,o_t) + \ell_t(o_t)$.
By summing over all times $t=1, \dotsc, T$, we get
\[\textstyle
\cost(\ftp) = \cost(A_T) \leq \cost(A_0) + 4\sum_{t=1}^T \eta_t,
\]
which equals $\off + 4\eta$.
\end{proof}

\subsection{Combining online algorithms}\label{sec:combining}
{We describe now how to make algorithm \ftp robust by combining it with a classical online algorithm. Although we only need to combine two algorithms, we will formulate the combination theorems more generally for any number of algorithms.}

Consider $m$ algorithms $A_0, \dotsc, A_{m-1}$ for some problem $P$
belonging to MTS.
We describe two methods to combine them into one algorithm which
achieves a performance guarantee close to the best of them.
Note that these methods are also applicable to problems which do not belong to
MTS as long as one can simulate all the algorithms at once and
bound the cost for switching between them.

\paragraph{Deterministic Combination.}
The following method was proposed by \citet{FiatRR94} for the
$k$-server problem, but can be generalized to MTS. We note that a
similar combination is also mentioned in \citet{LykourisV18}.
We simulate the execution of $A_0, \dotsc, A_{m-1}$ simultaneously.
At each time, we stay in the configuration of one of them,
and we switch between the algorithms in the manner of
a solution for the $m$-lane {\em cow path} problem,
see Algorithm~\ref{alg:comb-det} for details.

\begin{algorithm2e}
	\caption{$MIN^{det}$ \citep{FiatRR94}}
	\label{alg:comb-det}
	\DontPrintSemicolon
	choose $1 < \gamma \leq 2; \quad$ set $\ell := 0$\;
	\Repeat{the end of the input}{
		$i:= \ell \bmod m$\;%
		while $\cost(A_i) \leq \gamma^{\ell}$, follow $A_i$\;
		$\ell := \ell+1$\;
	}
\end{algorithm2e}

\begin{theorem}[generalization of Theorem 1 in~\citet{FiatRR94}]
	\label{thm:simulation_deterministic}
	Given $m$ online algorithms $A_0,\dots A_{m-1}$ for a problem $P$ in MTS, the
	algorithm $MIN^{det}$ achieves cost
	at most $(\frac{2\gamma^m}{\gamma-1}+1)\cdot \min_i \{cost_{A_i}(I)\}$,
	for any input sequence $I$.
\end{theorem}

A proof of this theorem can be found in Section~\ref{sec:simulation}.
The optimal choice of $\gamma$ is $\frac{m}{m-1}$. Then
$\frac{2\gamma^m}{\gamma-1}+1$ becomes 9 for $m=2$,
and can be bounded by $2em$ for larger $m$. Combined with Lemma~\ref{lem:ftp}, we obtain Theorem \ref{thm:ub-det}.

\paragraph{Randomized Combination.}
\citet{BlumB00} proposed the following way to combine online algorithms
based on the WMR~\citep{Littlestone1994} (Weighted Majority Randomized)
algorithm for the experts problem.
At each time $t$, it maintains a probability distribution $p^t$
over the $m$ algorithms updated using WMR.
Let $\dist(p^t, p^{t+1}) = \sum_i \max\{0, p_i^t - p_i^{t+1}\}$
be the earth-mover distance between $p^t$ and $p^{t+1}$
and let $\tau_{ij} \geq 0$ be the transfer of the probability mass from
$p^{t}_i$ to $p^{t+1}_j$ certifying this distance, so that
$p^{t}_i = \sum_{j=0}^{m-1} \tau_{ij}$ and
$\dist(p^t, p^{t+1}) = \sum_{i\neq j} \tau_{ij}$.
If we are now following algorithm $A_i$, we switch to $A_j$
with probability $\tau_{ij}/p^t_i$.
See Algorithm~\ref{alg:comb-rand} for details.
The parameter $D$ is an upper bound on the switching cost between the states
of two algorithms.

\begin{algorithm2e}
	\caption{$MIN^{rand}$ \citep{BlumB00}}
	\label{alg:comb-rand}
	$\beta := 1-\frac{\epsilon}2$\tcp*{for parameter $\epsilon < 1/2$}
	$w_i^0 := 1$ for each $i=0, \dotsc, m-1$\;
	\ForEach{time $t$}{
		$c_i^t:=$ cost incurred by $A_i$ at time $t$\;
		$w_i^{t+1} := w_i^t\cdot \beta^{c_i^t/D}$ and
		$p_i^{t+1} := \frac{w_i^{t+1}}{\sum w_i^{t+1}}$\;
		$\tau_{i,j} :=$ mass transferred from $p^t_i$ to $p^{t+1}_j$\;
		switch from $A_i$ to $A_j$ w.p. $\tau_{ij}/p_i^t$\;
	}
\end{algorithm2e}

\begin{theorem}[\citet{BlumB00}]
	\label{thm:simulation_randomized}
	Given $m$ on-line algorithms $A_0,\dots A_{m-1}$ for an MTS with diameter $D$
	and $\epsilon < 1/2$, there
	is a randomized algorithm $MIN^{rand}$ such that, for any instance $I$,
	its expected cost is at most
	\[ (1+\epsilon)\cdot \min_i \{\cost(A_i(I))\} + O(D/\epsilon)\ln m. \]
\end{theorem}
Combined with Lemma~\ref{lem:ftp}, we obtain Theorem~\ref{thm:ub-rand}.

\subsection{Lower bound}
\label{sec:lowerbounds}

We show that our upper bounds for general metrical task systems (Theorems \ref{thm:ub-det} and \ref{thm:ub-rand})
are tight up to constant factors.
We show this for MTS on a uniform metric, i.e., the metric where the
distance between any two points is $1$.

\begin{theorem}\label{thm:generalLB}
For $\bar \eta\ge 0$ and $n\in\N$, every deterministic (or randomized) online algorithm for MTS on the $n$-point uniform metric with access to an action prediction oracle with error at most $\bar\eta\cdot\opt$ with respect to some optimal offline algorithm has competitive ratio $\Omega\left(\min\left\{\alpha_n, 1+\bar{\eta}\right\}\right)$, where $\alpha_n=\Theta(n)$ (or $\alpha_n=\Theta(\log n)$) is the optimal competitive ratio of deterministic (or randomized) algorithms without prediction.
\end{theorem}
\begin{proof}
	For deterministic algorithms, we construct an input sequence consisting of phases defined as follows. We will ensure that the online and offline algorithms are located at the same point at the beginning of a phase. The first $\min\{n-2,\lfloor\bar\eta\rfloor\}$ cost functions of a phase always take value $\infty$ at the old position of the online algorithm and value $0$ elsewhere, thus forcing the algorithm to move. Let $p$ be a point that the online algorithm has not visited since the beginning of the phase. Only one more cost function will be issued to conclude the phase, which takes value $0$ at $p$ and  $\infty$ elsewhere, hence forcing both the online and offline algorithms to $p$. The optimal offline algorithm suffers a cost of exactly $1$ per phase because it can move to $p$ already at the beginning of the phase. The error is at most $\bar\eta$ per phase provided that point $p$ is predicted at the last step of the phase, simply because there are only at most $\bar\eta$ other steps in the phase, each of which can contribute at most $1$ to the error. Thus, the total error is at most $\bar\eta\opt$. The online algorithm suffers a cost $\min\{n-1,1+\lceil\bar\eta\rceil\}$ during each phase, which proves the deterministic lower bound.
	
	For randomized algorithms, let $k:=\lfloor\log_2 n\rfloor$ and fix a subset $F_0$ of the metric space of $2^k$ points. We construct again an input sequence consisting of phases: For $i=1,\dots,\min\{k,\lfloor\bar\eta\rfloor\}$, the $i$th cost function of a phase takes value $0$ on some set $F_i$ of feasible states and $\infty$ outside of $F_i$. Here, we define $F_i\subset F_{i-1}$ to be the set consisting of the half of the points of $F_{i-1}$ where the algorithm's probability of residing is smallest right before the $i$th cost function of the phase is issued (breaking ties arbitrarily). Thus, the probability of the algorithm already residing at a point from $F_i$ when the $i$th cost function arrives is at most $1/2$, and hence the expected cost per step is at least $1/2$. We assume that $\bar\eta\ge 1$ (otherwise the theorem is trivial). Similarly to the deterministic case, the phase concludes with one more cost function that forces the online and offline algorithms to some point $p$ in the final set $F_i$. Again, the optimal cost is exactly $1$ per phase, the error is at most $\bar\eta$ in each phase provided the last prediction of the phase is correct, and the algorithm's expected cost per phase is at least $\frac{1}{2}\min\{k,\lfloor\bar\eta\rfloor\}=\Omega(\min(\log n,1+\bar\eta))$, concluding the proof.
\end{proof}

In light of the previous theorem it may seem surprising that our algorithm \algname for caching (see Section~\ref{sec:paging}) achieves a competitive ratio logarithmic rather than linear in the prediction error, especially considering that the special case of caching when there are only $k+1$ distinct pages corresponds to an MTS on the uniform metric. However, the construction of the randomized lower bound in Theorem~\ref{thm:generalLB} requires cost functions that take value $\infty$ at several points at once, whereas in caching only one page is requested per time step.

\section{{Beyond Metrical Task Systems}}%
\label{sec:matching}

{The objective of this section is to show that the prediction setup
introduced in this paper is not limited to Metrical Task Systems, but can also
be useful for relevant problems not known to be inside this class. This
emphasizes the generality of our approach, compared to prediction setups
designed for a single problem. We focus on the \emph{online matching
on the line} problem, which has been studied for three decades and has seen recent
developments. }

In the \emph{online matching on the line} problem, we are given a set $S=\{s_1,s_2,\dots s_n\}$ of
server locations on the real line. A set of requests $R =
\{r_1,r_2,\dots r_n \}$ which are also locations on the real line,
arrive over time. Once request $r_i$ arrives, it has to be irrevocably
matched to some previously unmatched server $s_j$. The cost of this
edge in the matching is the distance between $r_i$ and $s_j$, i.e.,
$|s_j-r_i|$ and the total cost is given by the sum of all such edges
in the final matching, i.e., the matching that matches every request in $R$ to some unique
server in $S$. The objective is to minimize this total cost.

The best known lower bound on the competitive ratio of any
deterministic algorithm is $9.001$~\citep{FuchsHK05} and the best known
upper bound for any algorithm is $O(\log n)$, due
to~\citet{Raghvendra18}.

We start by defining the notion of \emph{distance} between two sets of servers.

\begin{definition}
Let $P_i^1$ and $P_i^2$ be two sets of points in a metric space, of
size $i$ each. We then say that their distance $\dist(P_i^1,P_i^2)$ is
equal to the cost of a minimum-cost perfect matching in the bipartite
graph having $P_i^1$ and $P_i^2$ as the two sides of the bipartition.
\end{definition}

In \emph{online matching on the line with action predictions} we assume that, in
each round $i$ along with request $r_i$, we obtain a prediction
$P_i\subseteq S$ with $|P_i|=i$ on the server set that the offline
optimal algorithm is using for the first $i$ many requests. We allow here even that $P_i\not\subseteq P_{i+1}$. 
The error in round $i$ is given by $\eta_i:= dist(P_i,\off_i)$, where
$\off_i$ is the server set of a (fixed) offline algorithm on the instance.
The total prediction error
is $\eta = \sum_{i=1}^n\eta_i$.

Since a request has to be irrevocably matched to a server, it is not
straightforward that one can switch between configurations of
different algorithms. Nevertheless, we are able to simulate such a
switching procedure. By applying this switching procedure to the best
known classic online algorithm for the problem, due to
\citet{Raghvendra18}, and designing a Follow-The-Prediction
algorithm that achieves a
competitive ratio of $1+2\eta/\off$, we can  apply the combining method of
Theorem~\ref{thm:simulation_deterministic} to get the following result.

\begin{theorem*}[Restated Theorem~\ref{thm:matching}]
  There exists a deterministic algorithm for the online matching on the line
  problem with action predictions that attains a competitive ratio of
  \begin{align*}
    \min\{O(\log n), 9+\frac{8e\eta}{\off}\},
  \end{align*}
  for any offline algorithm \off.
\end{theorem*}

We note that for some instances the switching cost between these two
algorithms (and therefore, in a sense, also the metric space diameter) can be as high as $\Theta(\opt)$ which renders the
randomized combination uninteresting for this particular problem.

\subsection{A potential function}

We define the \emph{configuration} of an
algorithm at some point in time as the set of servers which are currently
matched to a request.

For each round of the algorithm, we define $S_i$ as the current configuration and $P_i$ as the predicted configuration, which verify $|S_i|=|P_i|= i$. We define a potential function after each round $i$ to be
$\Phi_i = \dist(S_i,P_i)$, and let 
$\mu_i$ be the associated matching between $S_i$ and $P_i$ that
realizes this distance, such that all servers in $S_i \cap P_i$ are matched to
themselves for zero cost.
We extend $\mu_i$ to the complete set of severs $S$ by setting $\mu_i(q)=q$ for
all $q\notin S_i\cup P_i$.
The intuition behind the potential function is
that after round $i$ one can simulate being in configuration $P_i$
instead of the actual configuration $S_i$, at an additional expense of
$\Phi_i$.

\subsection{Distance among different configurations}
The purpose of this section is to show that the distance among the
configurations of two algorithms is at most the sum of their current
costs. As we will see, this will imply that we can afford switching
between any two algorithms.

We continue by bounding the distance between any two algorithms as a
function of their costs.
\begin{lemma}
  \label{lem:distance} Consider two algorithms $A$ and $B$, and fix
the set of servers $S$ as well as the request sequence $R$. Let
$A_i$ and $B_i$ be the respective configurations of the algorithms
(i.e., currently matched servers) after serving the first $i$
requests of $R$ with servers from $S$. Furthermore, let $\opt_i^A$
(resp. $\opt_i^B$) be the optimal matching between $\{r_1,r_2,\dots r_i\}$
and $A_i$ (resp. $B_i$), and let $M_i^A$  (resp. $M_i^B$) be the
corresponding matching produced by $A$ (resp. $B$). Then:
  \begin{align*}
    \dist(A_i,B_i) &\le \cost(\opt_i^A) + \cost(\opt_i^B)\\
     &\le \cost(M_i^A) + \cost(M_i^B).
  \end{align*}
\end{lemma}

\begin{proof}
  The second inequality follows by the optimality of $\opt_i^A$ and
  $\opt_i^B$. For the first inequality 
let $s_j^A$
(resp. $s_j^B$) be the server matched to $r_j$ by $\opt_i^A$ (resp.
$\opt_i^B$), for all $j\in \{1,\dots, i\}$. 
Therefore, 
there exists a matching between $A_i$ and $B_i$ that matches for all
$j\in\{1,\dots,i\}$, $s_j^A$ to $s_j^B$ which has a total cost of 
  \begin{align*}
  \sum_{j=1}^i \dist(s_j^A-s_j^B) &\le \sum_{j=1}^i \dist(s_j^A
    - r_j) + \sum_{j=1}^i \dist(s_j^B-r_j) \\
    &= \cost(\opt_i^A) + \cost(\opt_i^B),
  \end{align*}
where the inequality follows by the triangle inequality.
By the definition of distance we have that $\dist(A_i,B_i) \le
\sum_{j=1}^i \dist(s_j^A-s_j^B)$, which concludes the proof.
\end{proof}

\subsection{Follow-The-Prediction}

Since \emph{online matching on the line} is not known to be in MTS, we start
by redefining the algorithm Follow-The-Prediction for this particular
problem. In essence, the algorithm virtually switches from predicted
configuration $P_i$ to predicted configuration $P_{i+1}$.

 Let $S_i$ be the actual set of servers used by
Follow-The-Prediction after round $i$. 
Follow-The-Prediction 
computes the
optimal matching among $P_{i+1}$ and the multiset $P_i\cup\{r_{i+1}\}$
which maps the elements of $P_{i+1} \cap P_i$ to themselves.
Note that if $r_{i+1} \in P_i$, then $P_i\cup\{r_{i+1}\}$ is a multiset
where $r_{i+1}$ occurs twice.
Such matching will match $r_{i+1}$ to some server $s\in P_{i+1}\setminus P_i$. 
Recall that $\mu_i$ is the minimum cost bipartite matching
between $S_i$ and $P_i$ extended by zero-cost edges to the whole set of
servers. Follow-The-Prediction matches $r_{i+1}$ to the
server $\mu_i(s)$, i.e., to the server to which $s$ is matched to under $\mu_i$.
We can show easily that $\mu(s) \notin S_i$.
Since $s\notin P_i$, there are two possibilities:
If $s\notin S_i$, then $\mu(s) = s \notin S_i$ by extension of $\mu_i$ to
elements which do not belong to $S_i$ nor $P_i$.
Otherwise, $s \in S_i \setminus P_i$ and, since $\mu_i$ matches all
the elements of $S_i \cap P_i$ to themselves,
we have $\mu(s) \in P_i \setminus S_i$.

\begin{theorem}
  Follow-The-Prediction has total matching cost at most $\off+2\eta$ and
  therefore the algorithm has a competitive ratio of
  \begin{align*}
    1 + 2\eta/\off
  \end{align*}
against any offline algorithm \off.
\end{theorem}
\begin{proof}

  The idea behind the proof is that, by paying the switching cost of
  $\Delta\Phi_i$ at each round, we can always virtually assume that we reside in
  configuration $P_{i}$. So whenever a new request $r_{i+1}$ and a
  new predicted configuration $P_{i+1}$
  arrive, we pay the costs for switching from $P_i$ to $P_{i+1}$
  and for matching $r_{i+1}$ to a server in $P_{i+1}$. 

We first show that, for every round $i$, we have:

\begin{align*}
  FtP_i + \Delta\Phi_i &\leq \dist(P_{i+1}, P_i\cup\{r_{i+1}\})\\
\Leftrightarrow ~~ \dist(r_{i+1},\mu(s)) + \Phi_{i+1}  & \leq \dist(P_{i+1}, P_i\cup\{r_{i+1}\}) +  \Phi_i.
\end{align*}

Note that for all $j$, $\Phi_j = \dist(S_j, P_j) = \dist(\bar S_j, \bar P_j)$,
where $\bar S_j$ and $\bar P_j$ denote the complements
of $S_j$ and $P_j$ respectively.

We have in addition $\dist(\bar S_i, \bar P_i) = \dist(\bar S_i \setminus
\{\mu_i(s)\}, \bar P_i \setminus\{s\}) + \dist(s, \mu_i(s))$ as
$s\notin P_i$ and $\mu_i(s)\notin S_i$, and $(s,\mu_i(s))$ is an edge
in the min-cost matching between $\bar S_i$ and $\bar P_i$. Note that $S_{i+1} = S_i \cup \{\mu_i(s)\}$ so $\bar S_i \setminus \{\mu_i(s)\} = \bar S_{i+1}$. Therefore, we get:

 $$\Phi_i = \dist(\bar S_i, \bar P_i) = \dist(\bar S_{i+1},  \bar P_i \setminus\{s\}) + \dist(s, \mu_i(s))= \dist( S_{i+1},  P_i \cup\{s\}) + \dist(s, \mu_i(s)).$$

In addition, we have $\dist(P_{i+1}, P_i\cup\{r_{i+1}\}) = \dist(s,r_{i+1}) +
\dist(P_{i+1}\setminus \{s\}, P_i)$ because by definition of $s$,  $s$ is
matched to $r_{i+1}$ in a minimum cost matching between $P_{i+1}$ and $P_i\cup\{r_{i+1}\}$.
Now, $s\notin P_i$, so $\dist(P_{i+1}\setminus \{s\}, P_i) = \dist(P_{i+1}, P_i\cup \{s\})$ as this is equivalent to adding a zero-length edge from $s$ to itself to the associated matching. Therefore, we get:

$$\dist(P_{i+1}, P_i\cup\{r_{i+1}\}) = \dist(s,r_{i+1}) + \dist(P_{i+1}, P_i\cup \{s\}).$$

Combining the results above, we obtain:

\begin{align*}
  FtP_i + \Delta\Phi_i &\leq \dist(P_{i+1}, P_i\cup\{r_{i+1}\})\\
\Leftrightarrow ~~ dist(r_{i+1},\mu_i(s)) + \Phi_{i+1}  & \leq \dist(P_{i+1}, P_i\cup\{r_{i+1}\}) +  \Phi_i\\
\Leftrightarrow ~~ 
\dist(r_{i+1},\mu_i(s)) &+ \dist(S_{i+1}, P_{i+1})
\\\leq
\dist(S_{i+1},  &P_i \cup\{s\}) + \dist(s, \mu_i(s)) + 
\dist(s,r_{i+1}) + \dist(P_{i+1}, P_i\cup \{s\}) 
\end{align*}

The last equation holds by the triangle inequality.

Finally, we bound $\dist(P_{i+1}, P_i\cup\{r_{i+1}\})$ using the
triangle inequality. In the following $\off_i$ refers to the
configuration of offline algorithm \off after the first $i$ requests
have been served.

\begin{align*}
 &\dist(P_{i+1},P_i\cup\{r_{i+1}\}) \\&\leq
 \dist(P_i\cup \{r_{i+1}\},\off_i\cup \{r_{i+1}\})
    + \dist(\off_i\cup \{r_{i+1}\}, \off_{i+1}) +
    \dist(\off_{i+1},P_{i+1})\\  &\leq
 \eta_i + |\off_i| + \eta_{i+1}.
\end{align*}

Summing up over all rounds, and using that $\Phi_1=\Phi_n=0$  completes the proof of the theorem.
\end{proof}

\subsection{The main theorem}

The goal of this subsection is to prove Theorem~\ref{thm:matching}.
\begin{proof}[Proof of Theorem~\ref{thm:matching}]
  The main idea behind the proof is to show that we can apply
  Theorem~\ref{thm:simulation_deterministic}
  and virtually simulate the two algorithms (Follow-The-Prediction and the
  online algorithm of \citet{Raghvendra18}). 

  We need to show that we can assume that we are in some configuration
  and executing the respective algorithm, and that the switching cost
  between these configurations is upper bounded by the cost of the
  two algorithms. Similarly to the analysis of Follow-The-Prediction,
  we can virtually be in any configuration as long as we pay
  for the distance between any two consecutive configurations. When we
  currently simulate an algorithm $A$, the distance between the two
  consecutive configurations is exactly the cost of the edge that $A$
  introduces in this round. When we switch from the configuration of
  some algorithm $A$ to the configuration of some algorithm $B$, then
  by Lemma~\ref{lem:distance}, the distance between the two
  configurations is at most the total current cost of $A$ and $B$.
  
  This along with Theorem~\ref{thm:simulation_deterministic_appendix}
  (which is generalizing Theorem~\ref{thm:simulation_deterministic}
  beyond MTS and can be found in Appendix~\ref{sec:simulation}) concludes the proof.
 \end{proof}

\subsection{Bipartite metric matching}
\emph{Bipartite metric matching} is the generalization of online
matching on the line where the servers and requests can be points of
any metric space. The problem is known to have a tight
$(2n-1)$-competitive algorithm, due to
\citet{KalyanasundaramP93} as well as \citet{KhullerMV94}.

We note that our arguments in this section are not line-specific and
apply to that problem as well. This gives the following result:
\begin{theorem}
  There exists a deterministic algorithm for the online metric
  bipartite matching 
  problem with action predictions that attains a competitive ratio of
  \begin{align*}
    \min\{2n-1, 9+\frac{8e\eta}{\off}\},
  \end{align*}
against any offline algorithm \off.
\end{theorem}

\section{Experiments}\label{sec:experiments}

We evaluate the practicality of our approach on real-world datasets for two MTS: \emph{caching} and \emph{ice cream} problem. The source code and datasets are available at GitHub\footnote{\url{https://github.com/adampolak/mts-with-predictions}}. Each experiment was run $10$ times and we report the mean competitive ratios. The maximum standard deviation we observed was of the order of $0.001$.

\subsection{The caching problem}

\paragraph{Datasets.}

For the sake of comparability, we used the same two datasets as~\citet{LykourisV18}. 

\newcommand{\bk}{\texttt{BK}}
\newcommand{\citi}{\texttt{Citi}}

\begin{itemize}
	\item \bk\ dataset comes from a former social network BrightKite~\citep{BKpaper}.
  It contains checkins with user IDs and locations. We treat the sequence of checkin locations of each users as a separate instance of caching problem. We filter users with the maximum sequence length ($2100$) who require at least $50$ evictions in an optimum cache policy. Out of those we take the first $100$ instances.
  We set the cache size to $k=10$.
	
	\item \citi\ dataset comes from a bike sharing platform~\citet{citibike}.
  For each month of 2017, we consider the first $25\,000$ bike trips and build an
  instance where a request corresponds to the starting station of a trip.
  We set the cache size to $k=100$.
	
\end{itemize}

\paragraph{Predictions.}

We first generate the reoccurrence time predictions,
these predictions being used by previous
prediction-augmented algorithms. To this purpose, we use the same two
predictors as~\citet{LykourisV18}. Additionally we also
consider a simple predictor, which we call POPU (from \emph{popularity}),
and the LRU heuristic adapted to serve as a predictor.

\begin{itemize}
	\item Synthetic predictions: we first compute the exact reoccurrence time for each request, setting it to the end of the instance if it
does not reappear. We then add some noise drawn from a lognormal
distribution, with the mean parameter $0$ and the standard deviation $\sigma$, in
order to model rare but large failures. 

	\item PLECO predictions: we use the PLECO model described
in~\citet{PLECO}, with the same parameters as~\citet{LykourisV18},
which were fitted for the \texttt{BK} dataset (but not refitted for \texttt{Citi}).
This model estimates that a page requested $x$ steps earlier will be
the next request with a probability proportional to
$(x+10)^{-1.8}e^{-x/670}$. We sum the weights corresponding to all the earlier 
appearances of the current request to obtain the probability $p$ that this request is
also the next one. We then estimate that such a
request will reappear $1/p$ steps later.

	\item POPU predictions: if the current request has been seen in a
fraction $p$ of the past requests, we predict it will be repeated
$1/p$ steps later.

	\item LRU predictions: \citet{LykourisV18} already remarked on (but did not evaluate experimentally) a predictor that emulates the behavior of the LRU heuristic. A page requested at time $t$ is predicted to appear at time $-t$. Note that the algorithms only consider the order of predicted times among pages, and not their values, so the negative predictions pointing to the past are not an issue.

\end{itemize}

We then
transform the reoccurrence time predictions
to action predictions
by simulating the algorithm that evicts the element predicted
to appear the furthest in the future. In each step the prediction to
our algorithm is the configuration of this algorithm. %
Note that in the case of LRU predictions, the predicted configuration is precisely the configuration of the LRU algorithm.

\paragraph{Algorithms.}
We considered the following algorithms, whose competitive ratios are reported in Table~\ref{tbl:algorithms}.
	Two online algorithms:
 the heuristic LRU, which is considered the gold standard for caching, and %
 the $O(\log k)$-competitive Marker~\citep{FiatRR94}. %
	Three robust algorithms from the literature using the ``next-arrival time'' predictions: 
 L\&V~\citep{LykourisV18}, %
 LMarker~\citep{Rohatgi20}, and %
 LNonMarker~\citep{Rohatgi20}. %
	Three  algorithms using the prediction setup which is the focus of this paper:
 \ftp, which naively follows the predicted state,
	RobustFtP, which is defined as $MIN^{rand}(\text{\ftp},\text{Marker})$, and is an instance of the general MTS algorithm described in Section~\ref{sec:mts}, and
  \algname, the caching algorithm described in Section~\ref{sec:paging}.%

\begin{table}[tbp]
\small\centering
\begin{tabular}{@{}llll@{}}
  \toprule
  Algorithm & Competitive ratio & Property & Reference \\
	\midrule
  LRU & $k$ & & \citep{SleatorT85} \\
  Marker & $O(\log k)$ & Robust & \citep{FiatKLMSY91} \\
  FtP & $1 + 4\frac{\eta}{\opt}$ & {C+S} & Lemma~\ref{lem:ftp} \\
  L\&V & $2+O(\min\{\sqrt{\frac{\eta'}{\opt}}, \log k\})$ & {C+S+R} & \citep{LykourisV18}\\
  RobustFtP & $(1+\epsilon)  \min\{1+4\frac{\eta}{\opt}, O(\log k)\}$  & C+S+R  & Theorem~\ref{thm:ub-rand} \\
  LMarker & $O(1+\min\{\log\frac{\eta'}{\opt},\log k)\}$ & C+S+R  & \citep{Rohatgi20} \\
  LNonMarker & $O(1+\min\{1,\frac{\eta'}{k\cdot \opt}\}\log k)$ &  C+S+R  & \citep{Rohatgi20} \\
  \algname & $O(\min\{1+ \log(1+\frac{\eta}{\opt}), \log k\})$ &  C+S+R  & Theorem~\ref{thm:paging} \\
  \bottomrule
\end{tabular}
\caption{Summary of caching algorithms evaluated in experiments. Note that
$\eta$ and $\eta'$ are different measures of prediction error, so their
functions should not be compared directly. {Properties C+S+R mean
Consistency, Smoothness, and Robustness respectively.}}
\label{tbl:algorithms}
\end{table}

We implemented the deterministic and randomized combination schemes
described in Section~\ref{sec:combining} with a subtlety for the caching
problem: we do not flush the whole cache when switching algorithms, but perform only a single eviction per page fault in the same way as described in Remark~\ref{rmk:lazy}.
We set the parameters to $\gamma=1+0.01$ and {$\epsilon=0.5$. These values, chosen from $\{0.001,0.01,0.1,0.5\}$}, happen to be consistently the best choice in all our experimental settings.

\begin{figure}[tbh]
\centering
	\resizebox{.6\textwidth}{!}{\input{caching_bk_paper.pgf}}
	\caption{Comparison of caching algorithms augmented with synthetic predictions on the \texttt{BK} dataset.}
	\label{fig:pagingplot}
\end{figure}

\begin{table}[bth]
	\centering{
\begin{tabular}{l@{\qquad}ccc @{\qquad}ccc}
	\toprule
	\emph{Dataset} & \multicolumn{3}{c}{\texttt{BK}} & \multicolumn{3}{c}{\texttt{Citi}}\\
	\midrule
	LRU         & \multicolumn{3}{c}{1.291} & \multicolumn{3}{c}{1.848} \\
	Marker      & \multicolumn{3}{c}{1.333} & \multicolumn{3}{c}{1.861} \\
	\midrule
	\emph{Predictions}  & PLECO &  POPU &   LRU & PLECO &  POPU &  LRU \\
	\midrule
	\ftp                & 2.081 & 1.707 & 1.291 & 2.277 & 1.734 & 1.848 \\
	L\&V                & 1.340 & 1.262 & 1.291 & 1.877 & 1.776 & 1.848 \\
	LMarker             & 1.337 & 1.264 & 1.291 & 1.876 & 1.780 & 1.848 \\
	LNonMarker          & 1.333 & 1.292 & 1.299 & 1.862 & 1.771 & 1.855 \\
	\textbf{RobustFtP}  & 1.338 & 1.316 & 1.297 & 1.862 & 1.831 & 1.849 \\
	\textbf{\algname}   & 1.292 & 1.276 & 1.291 & 1.847 & 1.775 & 1.849 \\
	\bottomrule
\end{tabular}
}
\caption{Competitive ratios of caching algorithms using PLECO, POPU, and LRU predictions on both datasets.}
\label{tbl:paging}
\end{table}

\paragraph{Results.}
For both datasets, for each algorithm and each prediction considered, we
computed the total number of page faults over all the instances and divided it by the optimal number in order to obtain a \emph{competitive ratio}.
Figure~\ref{fig:pagingplot} presents the performance of a selection of the
algorithms depending on the noise of synthetic predictions for the
\texttt{BK} dataset. We omit LMarker and LNonMarker for readability since they
perform no better than L\&V. 
This experiment {shows} that our algorithm \algname outperforms previous
prediction-based algorithms as well as LRU {on the \texttt{BK} dataset with such predictions}.
Figures~\ref{fig:cachingplotbk} and~\ref{fig:cachingplotciti}
present the performance of all algorithms on the \texttt{BK} and \texttt{Citi}
datasets, respectively.
 {On the \texttt{Citi} dataset (Figure~\ref{fig:cachingplotciti}), \ftp achieves very good results even with a noisy synthetic predictor, and therefore {RobustFtP} surpasses other guaranteed algorithms. LNonMarker presents better performance for noisy predictions than the other algorithms.}

In Table~\ref{tbl:paging} we
provide the results obtained on both datasets using PLECO, POPU, and LRU
predictions. We observe that PLECO predictions are not accurate enough
to allow previously known algorithms to improve over the Marker
algorithm.
This may be due to the sensitivity of this predictor to consecutive identical requests, which are irrelevant for the caching problem.
However, using the simple POPU predictions enables the
prediction-augmented algorithms to significantly improve their
performance compared to the classical online algorithms. Using \algname with
either of the predictions is however sufficient to get a performance similar or
better than LRU (and than all other alternatives, excepted for POPU
predictions on the BK dataset). RobustFtP, although being a very
generic algorithm with worse theoretical guarantees than \algname, achieves a
performance which is not that far from previously known algorithms.
Note that we did not use a prediction model tailored to our setup,
which suggests that even better results can be achieved. 
When we use the LRU heuristic as a predictor, all the prediction-augmented algorithms perform comparably to the bare LRU algorithm.
For \algname and RobustFTP, there is a theoretical guarantee that this must be the case: Since the prediction error with respect to LRU is $0$, these algorithms are $O(1)$-competitive against LRU. Thus, \algname achieves both the practical performance of LRU with an exponentially better worst-case guarantee than LRU. Note that \citet{LykourisV18} also discuss how their algorithm framework performs when using LRU predictions, but did not provide both of these theoretical guarantees simultaneously.

\begin{figure}
  \centering
  \input{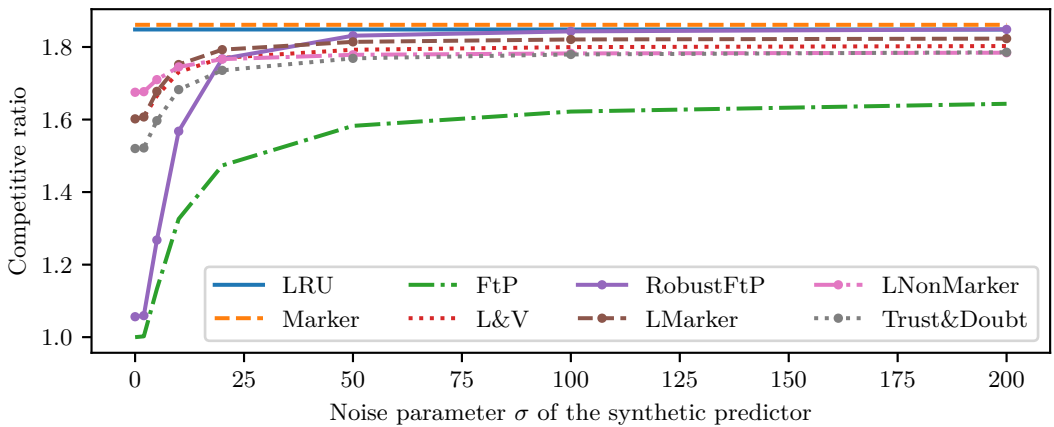}
  \caption{Comparison of caching algorithms augmented with synthetic predictions on the \texttt{BK} dataset.}
  \label{fig:cachingplotbk}
\end{figure}

\begin{figure}
  \centering
  \input{caching_citi_appendix.pgf}
  \caption{Comparison of caching algorithms augmented with synthetic predictions on the \texttt{Citi} dataset.}
  \label{fig:cachingplotciti}
\end{figure}

\subsection{A simple MTS: the \emph{ice cream} problem}

We consider a simple MTS example from \citet{chrobak98}, named \emph{ice cream} problem.
It it an MTS with two states, named $v$ and $c$, at distance $1$ from each other,
and two types of requests, $V$ and $C$. Serving a request while being in the matching state
costs $1$ for $V$ and $2$ for $C$, and the costs are doubled for the mismatched state.
The problem is motivated by an ice cream machine
which operates in two modes (states) -- vanilla or chocolate -- each facilitating a cheaper production of
a type of ice cream (requests).

We use the BrightKite dataset to prepare test instances for the problem. We extract the same $100$ users as for caching. For each user we look at the geographic coordinates of the checkins, and we issue a $V$ request for each checkin in the northmost half, and a $C$ request for each checkin in the southmost half.

In order to obtain synthetic predictions, we first compute the optimal offline policy, using dynamic programming. Then, for an error parameter $p$, for each request we follow the policy with probability $1-p$, and do the opposite with probability $p$.

We consider the following algorithms:
 the Work Function algorithm~\citep{BorodinLS92, BorodinEY1998}, of competitive
 ratio of $3$ in this setting ($2n-1$ in general);
 \ftp, defined in Section~\ref{sec:mts} (in case of ties in Equation~\ref{eq:ftp-def}, we follow the prediction); 
and the deterministic and randomized combination of the two above algorithms (with the same $\epsilon$ and $\gamma$ as previously) as proposed in Section~\ref{sec:mts}.

Figure~\ref{fig:MTS} presents the competitive ratios we obtained.
We can see that the general MTS algorithms we propose in Section~\ref{sec:mts}
allow to benefit from good
predictions while providing the worst-case guarantee of the classical online
algorithm. The deterministic and randomized combinations are comparable to the best of the
algorithms combined, and improve upon them when both algorithms have a similar performance.

\begin{figure}
\centering
  \input{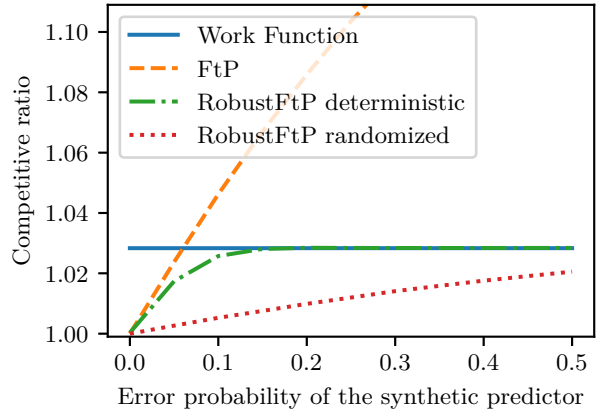}
  \caption{Performance on the ice cream problem with synthetic predictions.}
	\label{fig:MTS}
\end{figure}

\section{Conclusion}

In this paper, we proposed a prediction setup that allowed
us to design a general prediction-augmented algorithm for a large class of problems encompassing
MTS. For the MTS problem of caching in particular, the setup requires less information from the predictor
than previously studied ones {(since previous predictions can be converted to ours). Despite the more general setup,} we can design a specific algorithm
for the caching problem in our setup which offers {guarantees of a similar flavor to previous
algorithms} and even performs better in most of our experiments.

It may be considered somewhat surprising that a better bound is attainable for
caching than for general MTS, given that our lower bound instance for MTS uses
a uniform metric (and caching with a $(k+1)$-point universe also corresponds to a uniform metric). We conjecture logarithmic smoothness guarantees are
also attainable for other MTS problems with a request structure
similar to caching, like weighted caching and the $k$-server problem.
Further special cases of MTS can be obtained by restricting the number of
possible distinct requests (for example an MTS with two different possible
requests can model an important power management problem \cite{IraniSG03}),
or requiring a
specific structure from the metric space. Several such parametrizations of MTS
were considered by~\citet{BubeckR20} and it would be interesting to study
whether an improved dependence on the prediction error can be obtained in such
settings.

With respect to matching problems, there have been recent investigations through the lens of learning augmentation under specific matroid constraints in~\citep{AGKK20}, but this territory is still largely unexplored.
{It would also be interesting to evaluate our resource augmented algorithm for online metric matchings in real-life situations. As an example online matching algorithms are employed in several cities in order to match cars to parking spots (for example SFpark in San Francisco or ParkPlus in Calgary). Not only have such matching problems been studied from an algorithmic point of view (see, e.g., \citet{BenderG0P20}, \citet{BenderGP21}), but arguably it should be doable to generate high-quality predictions from historical data making our approach very promising.}

Another research direction is to identify more sophisticated predictors for caching and other problems that will further enhance the performance of prediction-augmented algorithms.

\bibliographystyle{abbrvnat}
\bibliography{draft}

\appendix

\section{Deterministic combination of a collection of algorithms}
\label{sec:simulation}

We consider a problem $P$ and $m$ algorithms $A_0, A_1, \dotsc, A_{m-1}$
for this problem which fulfill the following requirements.
\begin{itemize}
\item $A_0, \dotsc, A_{m-1}$ start at the same state and we are able
	to simulate the run of all of them simultaneously
\item for two algorithms $A_i$ and $A_j$, the cost of switching between
	their states is bounded by $\cost(A_i) + \cost(A_j)$.
\end{itemize}

\begin{theorem}[Restated Theorem~\ref{thm:simulation_deterministic};
generalization of Theorem 1 in~\citet{FiatRR94}]
\label{thm:simulation_deterministic_appendix}
Given $m$ on-line algorithms $A_0,\dots A_{m-1}$ for a problem $P$ which
satisfy the requirements above,
the algorithm $MIN^{det}$ with parameter $1< \gamma \leq 2$ incurs cost at most
\[ \left(\frac{2\gamma^{m}}{\gamma-1} + 1\right)
	\cdot \min_i\{ \cost(A_i(I))\}, \]
on any input instance $I$ such that $\opt_I \geq 1$.
If we choose $\gamma = \frac{m}{m-1}$, the coefficient
$\frac{2\gamma^{m}}{\gamma-1} + 1$ equals $9$ if $m=2$ and can be bounded
by $2em$.
\end{theorem}
Note that assumption on $\opt_I \geq 1$ is just to take care of the corner-case
instances with very small costs.
If we can only assume $\opt_I \geq c$ for some $0 < c < 1$,
then we scale all the costs fed to $MIN^{det}$ by $1/c$ and instances
with $\opt_I = 0$ are usually not very interesting.
The value of $c$
is usually clear from the particular problem in hand, e.g., for caching we only
care about instances which need at least one page fault, i.e., $\opt_I \geq 1$.
\begin{proof}
Let us consider the $\ell$-th cycle of the algorithm and denote
$i = \ell \bmod m$ and $i' = (\ell-1) \bmod m$.
We are switching from algorithm $A_{i'}$, whose current cost
we denote $\cost'(A_{i'}) = \gamma^{\ell-1}$ to $A_i$,
whose current cost we denote $\cost'(A_i)$,
and its cost at the end of this cycle will become
$\cost(A_i) = \gamma^\ell$. Our cost during this cycle, i.e.,
for switching and for execution of $A_i$, is at most
\[ \cost'(A_{i'}) + \cost'(A_i) + (\cost(A_i) - \cost'(A_i))
	= \cost'(A_{i'}) + \cost(A_i) = \gamma^{\ell-1} + \gamma^\ell.
\]

Now, let us consider the last cycle $L$, when we run the algorithm number
$i = L \bmod m$. By the preceding equation, the total cost of
$MIN^{det}$ can be bounded as
\[
\cost(MIN^{det})
	\leq 2 \cdot\sum_{\ell}^{L-1} \gamma^\ell + \cost(A_i)
	= 2 \frac{\gamma^L - 1}{\gamma -1} + \cost(A_i)
        \leq 2 \frac{\gamma^L}{\gamma-1} + \cost(A_i).
\]

If $L < m$, we use the fact that $\opt \geq 1$ and therefore the cost
of each algorithm processing the whole instance would be at least one.
Therefore, we have
\[ \cost(MIN^{det})
	\leq 2 \frac{\gamma^L}{\gamma-1} + \gamma^L
	\leq 2 \frac{\gamma^{m}}{\gamma-1} \cdot \min_i\{ \cost(A_i)\},
\]
because $\frac{\gamma^L}{\gamma-1} + \gamma^L = \frac{\gamma^{L+1}}{\gamma-1}$
and $L+1 \leq m$.

Now, we have $L \geq m$, denoting $i=L\bmod m$, and we distinguish two cases.

(1) If $\min_j\{ \cost(A_j)\} = \cost(A_i)$,
then $\cost(A_i) \geq \gamma^{L-m}$
for each $i$, and therefore
\[ \frac{\cost(MIN^{det})}{\min_j\{ \cost(A_j)\}}
	\leq \frac{2 \frac{\gamma^L}{\gamma-1} + \cost(A_i)}
		{\min_j\{ \cost(A_j)\}}
\]
Note that $\cost(A_i) \geq \gamma^{L-m}$, its cost from the previous usage.
Since $\min_j\{ \cost(A_j)\} = \cost(A_i)$, we get
\[ \frac{\cost(MIN^{det})}{\min_j\{ \cost(A_j)\}}
	\leq 2\frac{\gamma^{m}}{\gamma-1} + 1.
\]

(2) Otherwise, we have $\min_j\{ \cost(A_j)\} \geq \gamma^{L-m+1}$
and $\cost(A_j) \leq \gamma^L$ and therefore
\[ \frac{\cost(MIN^{det})}{\min_j\{ \cost(A_j)\}}
	\leq 2\frac{\gamma^{m-1}}{\gamma-1} + \gamma^{m-1}
	\leq 2\frac{\gamma^{m}}{\gamma-1}.
\]

For $\gamma = \frac{m}{m-1}$ we have
\[ 2\frac{\gamma^{m}}{\gamma-1} + 1 = 2 (m-1) \left(\frac{m}{m-1}\right)^{m}+1,
\]
which equals $9$ for $m=2$ and can be bounded by $2em$.

\end{proof}

\section{Comparison between \algname and the best marking algorithm}
\label{sec:TDvsMarking}

The algorithm \algname does not belong to the broad class of \emph{marking}
algorithms. We notice in this section that, given perfect predictions, this
property allows it to outperform all marking algorithms on some instances, but,
 at the same time, it does not always perform as well as the best marking
algorithm even when given perfect predictions.

\begin{remark}
With perfectly accurate predictions, there exist both a caching instance on
which \algname performs better than the best marking algorithm, and another
caching instance on which \algname is outperformed by a marking algorithm.
\end{remark}

\begin{proof}
We first build an instance where \algname, given predictions corresponding to
the optimal algorithm evicting the page arriving the furthest in the future,
outperforms the best marking algorithm. Consider a cache of size 3 and the
request sequence 1, 2, 3; 4, 5, 6; 1, 2, 3, composed of three phases of length
three (separated by semicolons). \algname keeps the pages 1 and 2 in cache
during the second phase so suffers seven cache misses. The best marking
algorithm is not able to keep such old pages in cache so suffers nine cache misses.

Now, we build an instance where \algname, given again predictions corresponding to
the optimal algorithm evicting the page arriving the furthest in the future, suffers
more cache misses than the best marking algorithm. Consider a cache of size 3,
and the request sequence 1, 2, 3; 4, 5, 6, 5, 6; 7, 1, 4, composed of three
phases of length three, five and three. The best marking algorithm suffers eight cache
misses, the page 4 being present in the cache for the last request. The cache of
\algname after the second phase contains 1, 5, 6, as the page 1 is given
priority over the page 4, and, at the start of the last phase, the now ancient
page 1 is evicted, so the algorithm suffers nine cache misses.
\end{proof}

\section{Limitations of the reoccurrence time predictions}
\label{sec:limitations}
In this section, we prove Theorem~\ref{thm:limitations}.
In previous works on caching~\citep{LykourisV18,Rohatgi20,AlexWei20},
the predictions are the time of the next reoccurrence to each page.
It is natural to try extending this type of predictions to
other problems, such as weighted caching. In weighted caching each page
has a weight/cost that is paid each time the page enters the
cache. However, it turns out that even with perfect predictions of
this type for weighted caching, one cannot improve upon the competitive
ratio $\Theta(\log k)$, which can already be attained without predictions~\citep{BansalBN12}. Our proof is based on a known lower bound for MTS on a so-called ``superincreasing'' metric \citep{KarloffRR94}. Following a presentation of this lower bound by~\cite{Lee18}, we modify the lower bound so that the perfect predictions provide no additional information.

We call an algorithm for weighted caching \emph{semi-online} if it is
online except that it receives in each time step, as an additional input, the
reoccurrence time of the currently requested page (guaranteed to be without
error). We prove the following result:

\begin{theorem*}[Restated Theorem~\ref{thm:limitations}]
Every randomized semi-online algorithm for weighted caching is $\Omega(\log
k)$-competitive.
\end{theorem*}
\begin{proof}
Let $\tau>0$ be some large constant. Consider an instance of weighted caching with cache size $k$ and $k+1$ pages, denoted by the numbers $0,\dots,k$, and such that the weight of page $i$ is $2\tau^i$. It is somewhat easier to think of the following equivalent \emph{evader} problem: Let $S_k$ be the weighted star with leaves $0,1,\dots,k$ and such that leaf $i$ is at distance $\tau^i$ from the root. A single evader is located in the metric space. Whenever there is a request to page $i$, the evader must be located at some leaf of $S_k$ other than $i$. The cost is the distance traveled by the evader. Any weighted caching algorithm gives rise to the evader algorithm that keeps its evader at the one leaf that is \emph{not} currently in the algorithm's cache. The cost between the two models differs only by an additive constant (depending on $k$ and $\tau$).

For $h=1,\dots,k$ and a non-empty time interval $(a,b)$, we will define inductively a random sequence $\sigma_h=\sigma_h(a,b)$ of requests to the leaves $0,\dots,h$, such that each request arrives in the time interval $(a,b)$ and
\begin{align}\label{eq:randLbSemiOnline}
A_h\ge4\alpha_{h-1}\tau^{h}\ge \alpha_h\cdot \opt_h,
\end{align}
where $A_h$ denotes the expected cost of an arbitrary semi-online algorithm to serve the random sequence $\sigma_h$ while staying among the leaves $0,\dots,h$, $\opt_h$ denotes the expected optimal offline cost of doing so with an offline evader that starts and ends at leaf $0$, $\alpha_0=\frac{1+\tau}{4\tau}$, and $\alpha_h=1/2+\beta\log h$ for $h\ge 1$, where $\beta>0$ is a constant determined later. The inequality between the first and last term in \eqref{eq:randLbSemiOnline} implies the theorem. We will also ensure that $(0,1,\dots,h)$ is both a prefix and a suffix of the sequence of requests in $\sigma_h$.

For the base case $h=1$, the inequality is satisfied by the request sequence $\sigma_1$ that requests first $0$ and then $1$ at arbitrary times within the interval $(a,b)$.

For $h\ge 2$, the request sequence $\sigma_{h}$ consists of subsequences (iterations) of the following two types (we will only describe the sequence of request locations for now and later how to choose the exact arrival times of these requests): A \emph{type 1 iteration} is the sequence $(0,1,\dots,h)$. A \emph{type 2 iteration} is the concatenation of $\lceil\frac{\tau^{h}}{\alpha_{h-1}\opt_{h-1}}\rceil$ independent samples of a random request sequence of the form $\sigma_{h-1}$. The request sequence $\sigma_{h}$ is formed by concatenating $\lceil8\alpha_{h-1}\rceil$ iterations, where each iteration is chosen uniformly at random to be of type 1 or type 2. If the last iteration is of type 2, an additional final request at $h$ is issued. Thus, by induction, $(0,\dots,h)$ is both a prefix and a suffix of $\sigma_h$.

We next show \eqref{eq:randLbSemiOnline} under the assumption that at the start of each iteration, the iteration is of type 1 or 2 each with probability $1/2$ even when conditioned on the knowledge of the semi-online algorithm at that time. We will later show how to design the arrival times of
individual requests so that this assumption is satisfied. We begin by proving the first inequality of \eqref{eq:randLbSemiOnline}. We claim that in each iteration of $\sigma_h$, the expected cost of any semi-online algorithm (restricted to staying at the leaves $0,\dots,h$) is at least $\tau^{h}/2$. Indeed, if the evader starts the iteration at leaf $h$, then with probability $1/2$ we have a type $1$ iteration forcing the evader to vacate leaf $h$ for cost $\tau^h$, giving an expected cost of $\tau^h/2$. If the evader is at one of the leaves $0,\dots,h-1$, then with probability $1/2$ we have a type 2 iteration. In this case, it must either move to $h$ for cost at least $\tau^{h}$, or $\lceil\frac{\tau^{h}}{\alpha_{h-1}\opt_{h-1}}\rceil$ times it suffers expected cost at least $\alpha_{h-1}\opt_{h-1}$ by the induction hypothesis. So again, the expected cost is at least $\tau^{h}/2$. Since $\sigma_{h}$ consists of $\lceil8\alpha_{h-1}\rceil$ iterations, we have
\begin{align*}
A_h\ge 4\alpha_{h-1}\tau^{h},
\end{align*}
giving the first inequality of \eqref{eq:randLbSemiOnline}.

To show the second inequality of \eqref{eq:randLbSemiOnline}, we describe an offline strategy. With probability $2^{-\lceil8\alpha_{h-1}\rceil}$, all iterations of $\sigma_h$ are of type 2. In this case, the offline evader moves to leaf $h$ at the beginning of $\sigma_h$ and back to leaf $0$ upon the one request to $h$ at the end of $\sigma_h$, for total cost $2(1+\tau^h)$. With the remaining probability, there is at least one type 1 iteration. Conditioned on this being the case, the expected number of type 1 iterations is $\lceil8\alpha_{h-1}+1\rceil/2$, and the expected number of type 2 iterations is $\lceil8\alpha_{h-1}-1\rceil/2$. The offline evader can serve each type 1 iteration for cost $2(1+\tau)$ and each type 2 iteration for expected cost $\lceil\frac{\tau^{h}}{\alpha_{h-1}\opt_{h-1}}\rceil\opt_{h-1}$, and it finishes each iteration at leaf $0$. (Thus, if the last iteration is of type 2, then the final request to $h$ incurs no additional cost.) By the induction hypothesis, $\opt_{h-1}\le O(\tau^{h-1})$. Hence, we can rewrite the expected cost of a type 2 iteration as
\begin{align*}
\left\lceil\frac{\tau^{h}}{\alpha_{h-1}\opt_{h-1}}\right\rceil\opt_{h-1}= (1+o(1))\frac{\tau^h}{\alpha_{h-1}},
\end{align*}
as $\tau\to\infty$. Since $h\ge 2$, the expected cost of all type 1 iterations is only an $o(1)$ fraction of the expected cost of the type 2 iterations.
Overall, we get
\begin{align*}
\opt_h&\le 2^{-\lceil8\alpha_{h-1}\rceil}2(1+\tau^h) +\\ &~~~~~(1+o(1))\left(1-2^{-\lceil8\alpha_{h-1}\rceil}\right)\frac{\lceil8\alpha_{h-1}-1\rceil}{2}\frac{\tau^h}{\alpha_{h-1}}\\
&\le (1+o(1))\left[2^{-\lceil8\alpha_{h-1}\rceil}2\tau^h + \left(1-2^{-\lceil8\alpha_{h-1}\rceil}\right)4\tau^h\right]\\
&= (1+o(1))\left(1-2^{-\lceil8\alpha_{h-1}\rceil-1}\right)4\tau^h\\
&\le \frac{4\tau^h}{1+2^{-\lceil8\alpha_{h-1}\rceil-1}}.
\end{align*}
We obtain the second inequality in \eqref{eq:randLbSemiOnline} by
\begin{align*}
\frac{4\alpha_{h-1}\tau^{h}}{\opt_h}&\ge \alpha_{h-1}\left(1+2^{-\lceil8\alpha_{h-1}\rceil-1}\right)\\
&\ge \frac{1}{2}+\beta\log (h-1)+2^{-8\beta\log(h-1)-7}\\
&\ge \frac{1}{2}+\beta\log (h-1)+\frac{\beta}{h-1}\\
&\ge \frac{1}{2}+\beta\log h\\
&=\alpha_h,
\end{align*}
where the third inequality holds for $\beta=2^{-7}$.

It remains to define to define the arrival times for the requests of
sequence $\sigma_h$ within the interval $(a,b)$. We do this as
follows: Let $m\ge 1$ be the number of requests to leaf $h$ in
$\sigma_h$. These requests to $h$ will be issued at times
$a+(b-a)\sum_{i=1}^j 2^{-i}$ for $j=1,\dots,m$.

To define the arrival times of the other requests, we will maintain a time variable $c\in[a,b)$ indicating the \emph{current} time, and a variable $n>c$ indicating the time of the \emph{next} request to leaf $h$ after time $c$. Initially, $c:=a$ and $n:=(a+b)/2$. Consider the first iteration for which the arrival times have not been defined yet. If the iteration is of type 2, we choose the arrival times according to the induction hypothesis so that all subsequences $\sigma_{h-1}$ within the iteration fit into the time window $(c,(c+n)/2)$, and we update $c:=(c+n)/2$. If the iteration is of type 1, sample a type 2 iteration $I$ and let $t_1,\dots,t_{h-1}$ be such that $t_i$ would be the time of the next request to page $i$ if the next iteration were this iteration $I$ of type 2 instead of a type 1 iteration. We define the arrival times of the (single) request to leaf $i<h$ in this type 1 iteration to be $t_i$. If this was not the last iteration, we update $c:=n$ and increase $n$ to the time of the next request to $h$ (as defined above).

Notice that at the beginning of each iteration within $\sigma_h$,
ordering the pages by the time of their next request always yields the
sequence $0,1,\dots,k$, and the time of the next request to each page
is independent of whether the next iteration is of type 1 or type
2. Thus, as promised, whether the next iteration is of type 1 or type 2 is independent of the knowledge of the semi-online algorithm.\end{proof}

\end{document}